\documentclass[11pt]{article}
\usepackage{amsmath}
\usepackage{amssymb}
\usepackage{amsthm}
\usepackage{mathdots}
\usepackage{tikz}
\usepackage{algorithm}
\usepackage{algorithmic}
\usetikzlibrary{automata}
\usetikzlibrary{positioning}
\usetikzlibrary{shapes.misc}
\usepackage[normalem]{ulem}
\usepackage{subcaption}
\newcommand{\defeq}{\overset{\rm def}{=}}

\newcommand{\ceil}[1]{\lceil{#1}\rceil}
\newcommand{\flip}{\overline}

\newcommand{\E}{\mathrm{E}}

\newcommand{\tent}{f}

\newcommand{\bitseq}{\mathbf{b}}
\newcommand{\bitseqc}{\mathbf{c}}
\newcommand{\bitseqd}{\mathbf{d}}
\newcommand{\enc}{\gamma}
\newcommand{\se}{S} 
\newcommand{\type}{Type}
\newcommand{\typei}{I}
\newcommand{\typej}{J}
\newcommand{\typefn}{T}

\newcommand{\typeset}{\mathcal{T}}

\newcommand{\partner}{\theta}

\newcommand{\requiredstgsize}{K}

\newcommand{\poly}{\mathrm{poly}}
\newcommand{\Order}{\mathrm{O}}
\newcommand{\order}{\mathrm{o}}
\newcommand{\lev}{L}
\newcommand{\outn}{N}

\newtheorem{theorem}{Theorem}[section]
\newtheorem{lemma}[theorem]{Lemma}
\newtheorem{corollary}[theorem]{Corollary}
\newtheorem{proposition}[theorem]{Proposition}

\newtheorem{observation}[theorem]{Observation}

\title{The Space Complexity of Generating Tent Codes}
\author{
 Naoaki Okada\footnote{Graduate School of Information Science and Electrical Engineering, Kyushu University} \and 
 Shuji Kijima\footnote{Faculty of Data Science, Shiga University}
}

\begin{document}
\maketitle

\begin{abstract}
 This paper is motivated by a question 
  whether it is  possible to calculate a chaotic sequence efficiently, 
  e.g., is it possible to get the $n$-th bit of a bit sequence generated by a chaotic map, 
   such as $\beta$-expansion, tent map and logistic map in $\order(n)$ time/space?
 This paper gives an affirmative answer to the question about the space complexity of a tent map. 
 We prove that a tent code of $n$-bits 
   with an initial condition uniformly at random 
  is exactly generated  in $\Order(\log^2 n)$ space in expectation. 
\end{abstract}

\section{Introduction}
A {\em tent map} $\tent_{\mu} \colon [0, 1] \to [0, 1]$ (or simply $\tent$) is given by 
\begin{align}\label{eq:tentmap}
  \tent(x) =
  \begin{cases}
    \mu x &: x \leq \frac{1}{2}, \\
    \mu (1 - x) &: x \ge \frac{1}{2}
  \end{cases}
\end{align}
  where this paper is concerned with the case of $1 < \mu < 2$. 
 As Figure~\ref{fig:tentmap} shows, 
   it is a simple piecewise-linear map looking like a tent.  
 Let $x_n = \tent(x_{n-1}) = \tent^n(x) $ recursively for $n=1,2,\ldots$, where $x_0=x$ for convenience. 
 Clearly, $x_0,x_1,x_2,\ldots$ is a deterministic sequence.
 Nevertheless, the deterministic sequence shows a complex behavior, as if ``random,'' when $\mu>1$. 
 It is said {\em chaotic}~\cite{Lorenz93}. 
 For instance, 
  $\tent^n(x)$ becomes quite different from $\tent^n(x')$ for $x \neq x'$ as $n$ increasing, 
    even if $|x-x'|$ is very small, and 
   it is one of the most significant characters of 
    a chaotic sequence   
    known as the  {\em sensitivity to initial conditions} --- 
    a chaotic sequence is ``unpredictable'' despite a deterministic process~\cite{Lorenz63,SMYO83,YMS83,CH94}. 

 This paper,  
   from the viewpoint of theoretical computer science, 
   is concerned with the computational complexity of a simple problem: 
 Given $\mu$, $x$ and $n$, 
  decide whether $\tent^n(x) < 1/2$. 
 Its time complexity might be one of 
  the most interesting questions; 
   e.g., is it possible to ``predict'' whether $\tent^n(x) < 1/2$ in time polynomial in $\log n$? 
 Unfortunately, we in this paper cannot answer the question\footnote{
   We think that the problem might be NP-hard 
    using the arguments on the complexity of algebra and number theory in~\cite{GJ79}, 
    but  we could not find the fact. 
  }. 
 Instead, 
   this paper 
   is concerned with the space complexity of the problem.  

\begin{figure}[tbp]
  \begin{tabular}{cc}
    \begin{minipage}[t]{.5\hsize}
      \centering
      \includegraphics[width=1\linewidth]{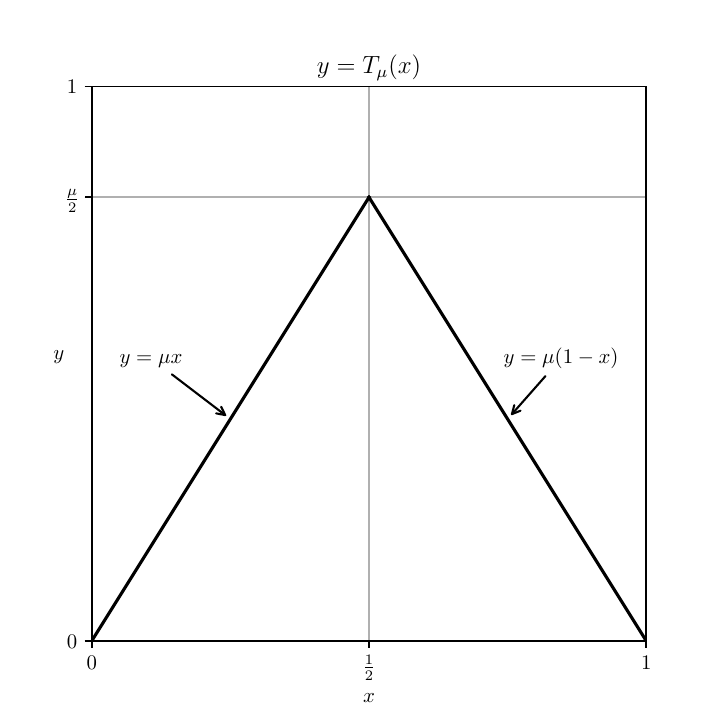}
      \subcaption{Tent map $\tent(x)$}
      \label{fig:tentmap}
    \end{minipage} &
    \begin{minipage}[t]{.5\hsize}
      \centering
      \includegraphics[width=1\linewidth]{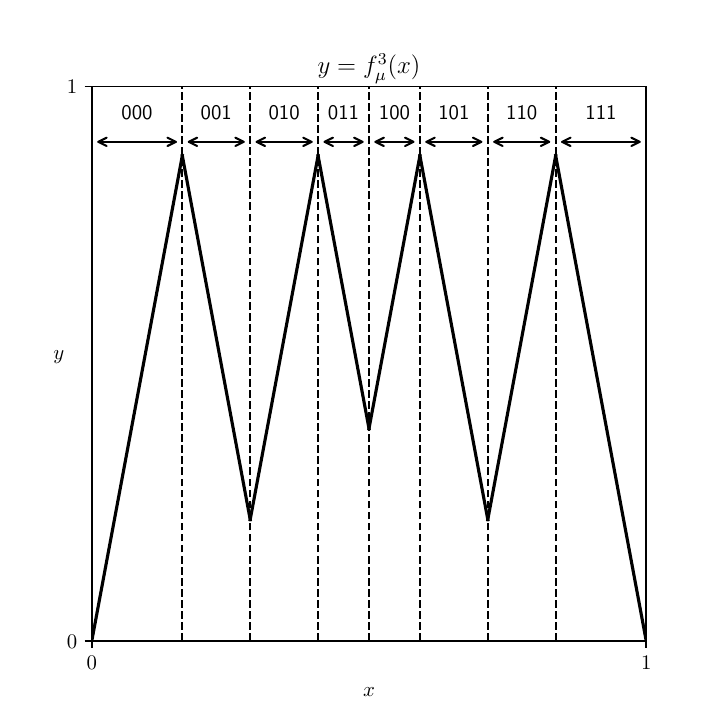}
      \subcaption{The $3$rd iterated tent map $\tent^3(x)$}
      \label{fig:expansion}
    \end{minipage}
  \end{tabular}
  \caption{The tent map $\tent(x)$ and the $3$-rd iterated tent map $\tent^3(x)$.}
  \label{fig:compress0}
\end{figure}

\subsection{Background and Contribution}
\paragraph{Chaos.}
 Chaotic sequences show many interesting figures 
  such as 
    cobweb, 
    strange attractor, 
    bifurcation, etc.~\cite{Lorenz63,LY75,May76,Lorenz93,Kohda08}. 
 The chaos theory has been intensively developed 
   in several context such as 
     electrical engineering, information theory, statistical physics, neuroscience and computer science,  
  with many applications, such as weather forecasting, traffic prediction, stock pricing,  since the 1960s. 
 For instance, 
   a cellular automaton, including the life game, is 
    a classical topic in computer science, and 
   it is closely related to the ``edge of chaos.'' 
 For another instance, 
  the ``sensitivity to initial conditions'' are often regarded as unpredictability, 
   and chaotic sequences are used in 
    pseudo random number generator,  
    cryptography, 
    or heuristics for NP-hard problems including chaotic genetic algorithms. 
 
 From the viewpoint of theoretical computer science, 
  the numerical issues of computing 
   chaotic sequences have been intensively investigated 
  in statistical physics, information theory and probability. 
 In contrast, 
   the computational complexity of computing a chaotic sequence 
    seems not well developed. 
 It may be a simple reason that 
  it looks unlike a decision problem. 

\paragraph{Tent map: 1-D, piecewise-linear and chaotic.}
 Interestingly, 
   very simple maps show chaotic behavior. 
 One of the most simplest maps are piece-wise linear maps, 
  including the tent map and the $\beta$-expansion (a.k.a. Bernoulli shift) which are 1-D maps and 
    the baker's map which is a 2-D map 
  \cite{Lorenz63,Renyi57,Parry60,Parry64,SMYO83,YMS83,CH94,Hopf37,baker00}.

 The tent map, as well as the $\beta$-expansion, is known to be 
  topologically conjugate to the {\em logistic map} 
   which is a quadratic map cerebrated as a chaotic map. 
 Chaotic behavior of the tent map, in terms of power spectra, band structure, critical behavior, are analyzed in e.g.,~\cite{Lorenz63,SMYO83,YMS83,CH94}. 
 The tent map is also used for pseudo random generator or encryption e.g., \cite{addabbo2006,Access20,LLQL17}. 
 It is also used for  meta-heuristics for NP-hard problems~\cite{DNA09,NoC09}.

\paragraph{Our results and related works.}
 This paper is concerned with the problem 
   related to deciding whether $\tent^n(x) < 1/2$ for the $n$-th iterated tent map $\tent^n$. 
 More precisely, 
   we define the {\em tent language} ${\cal L}_n \subseteq \{0,1\}^n$  in Section~\ref{sec:tent-map}, 
   and consider a {\em random generation} of $B \in {\cal L}_n$ according to a distribution ${\cal D}_n$
     corresponding to the ``uniform initial condition'' 
   (see Section~\ref{sec:tent-map}, for the precise definition). 
 We give an algorithm in $\Order(\log^2 n)$ {\em expected} space   (Theorem~\ref{thm:main}), 
  meaning that the {\em average space complexity} is exponentially improved 
    compared with the naive computation according to \eqref{eq:tentmap}. 

 Our strategy is as follows: 
  we 
    give a compact state-transit model for the tent language ${\cal L}_n$ in Section~\ref{sec:recog}, 
    design a random walk on the model to provide a desired distribution ${\cal D}_n$  in Section~\ref{sec:MC}, and then 
    prove that the expected space complexity is $\Order(\log^2 n)$ in Section~\ref{sec:average-complexity}. 
 The idea of the compact model and the random walk is similar to \cite{tomitamasterthesis,tomita2018}  for $\beta$-expansion, 
   while \cite{tomitamasterthesis,tomita2018}  did not give any argument on the space complexity beyond a trivial upper bound. 
 For the compact representation, 
   we use the idea of the {\em location of segments} 
     (which we call {\em segment-type} in this paper) developed by \cite{makino2015}. 

 Our technique is related to the {\em Markov partition}, often appearing in entropy arguments~\cite{sinai68,teramoto10}. 
 For some $\mu$, precisely when $\mu$ is a solution of $\tent^n(1/2)=1/2$ for an integer $n$, 
   our Markov chain is regarded as a version of Markov partition with a finite state. 
 However, the number of states of our Markov chain is unbounded in general as $n \to \infty$, and 
  it is our main target.

 For an earlier version of this manuscript, 
   Masato Tsujii gave us a comment about the connection to the {\em Markov extension}. 
 In 1979, Hofbauer~\cite{Hofbauer79} gave a representation of the kneading invariants for unimodal maps, 
  which is known as the Markov extension and/or Hofbauer tower, and then 
   discussed {\em topological entropy}. 
 Hofbauer and Keller extensively developed the arguments in 1980s, see e.g., \cite{deMelo-vanStrien,Bruin95}. 
 In fact, 
   some arguments of this manuscript, namely Sections ~\ref{sec:recog}, \ref{sec:MC}, and \ref{sec:transition_function}, 
    are very similar to or essentially the same as the arguments of the Markov extension, cf. \cite{deMelo-vanStrien}, 
  whereas 
   note that this manuscript is mainly concerned with the {\em computational complexity}.  
 It is difficult to describe the idea of our main result, 
     namely Algorithm~\ref{alg1} given in section \ref{sec:algo-summary} and Theorem \ref{thm:main}, 
   without describing those arguments, 
  then Sections \ref{sec:recog}, \ref{sec:MC} and \ref{sec:transition_function} are left as it is. 
 We also use some classical techniques of computational complexity, cf. \cite{Sipser,KV}, for the purpose.

\section{Tent Code and Main Theorem}\label{sec:tent-map}
 This paper focusing on the case  $\mu \in (1,2)$. 
 We assume that $\mu$ is rational in the main theorem (Theorem~\ref{thm:main}), 
  to make the arguments clear on the Turing computation, 
  but the assumption is not essential\footnote{
  See Section~\ref{sec:real} for real $\mu$. 
    }.  
Let $\tent^n$ denote the {\em $n$-times iterated tent map}, 
 which is formally given by $\tent^n(x)=\tent(\tent^{n-1}(x))$ recursively with respect to $n$. 
We remark, for the later argument in Section~\ref{sec:recog},  that 
\begin{align} 
 \tent^{n}(x)=\tent^{n-1}(\tent(x))
\label{eq:recursion}
\end{align}
  also holds by the associative law of a function composition. 
 It is easy to observe from the definition hat the iterated tent map is symmetric around $1/2$,  
  meaning that 
\begin{align}
 \tent^n(x) = \tent^n(1-x)
\label{eq:tent-sym}
\end{align} 
holds for any $x \in (0,1)$. 

%

 We define a tent-encoding function $\enc^{n}_{\mu}\colon [0,1) \to \{0,1\}^n$ (or simply $\enc^n$) as follows. 
For convenience, 
 let $x_i=\tent^i(x)$ for $i=1,2,\ldots$ as given $x \in [0, 1)$. 
Then, the {\em tent code} $\enc^{n}(x)=b_1 \cdots b_n$ for $x \in [0, 1)$ is a bit-sequence, 
  where 
\begin{align}
  b_{1} &=
  \begin{cases}
    0 &: x < \frac{1}{2}, \\
    1 &: x \ge \frac{1}{2},
  \end{cases} 
  \label{def:encode0} 
\end{align}
 and $b_i$ ($i=2,3,\ldots,n$)  is recursively given by 
\begin{align}
  b_{i+1} &=
  \begin{cases}
    b_{i} &: x_i < \frac{1}{2}, \\
    \flip{b_{i}} &: x_i > \frac{1}{2}, \\
    1 &: x_i = \frac{1}{2},
  \end{cases}
\label{def:encode1}
\end{align}
where $\flip{b}$ denotes bit inversion of $b$,
i.e., $\overline{0}=1$ and $\overline{1}=0$.
We remark that the definition \eqref{def:encode1} is rephrased by
\begin{align}
  b_{i+1} &=
  \begin{cases}
    0 &: [b_i = 0] \wedge \left[x_i < \frac{1}{2}\right], \\
    1 &: [b_i = 0] \wedge \left[x_i \geq \frac{1}{2}\right], \\
    1 &: [b_i = 1] \wedge \left[x_i \leq \frac{1}{2}\right], \\
    0 &: [b_i = 1] \wedge \left[x_i > \frac{1}{2}\right].
  \end{cases}
\label{def:encode2}
\end{align}

\begin{proposition}\label{prop:tent-expansion}
 Suppose $\enc_{\mu}^{\infty}(x) = b_1b_2\cdots$ for $x \in [0,1)$. 
 Then, $(\mu - 1) \sum_{i=1}^{\infty} b_{i} \mu^{-i} = x$. 
\end{proposition}

 See Section~\ref{apx:tent-map} for a proof of Proposition \ref{prop:tent-expansion}, 
  proofs of Propositions \ref{prop:order}, \ref{prop:-heikai} and Lemma~\ref{lem:encode1} below as well. 
 The proofs are not difficult but lengthy.  
 Thanks to this a little bit artificial definition \eqref{def:encode1},  
   we obtain the following two more facts. 
\begin{proposition}\label{prop:order}
  For any $x, x^{\prime} \in [0, 1)$,
  \begin{align*}
    x \leq x^{\prime} &\Rightarrow \enc_{\mu}^n(x) \preceq \enc_{\mu}^n(x^{\prime})
  \end{align*}
  hold where $\preceq$ denotes the {\em lexicographic order}, 
   that is $b_{i_*}=0$ and $b'_{i_*}=1$ at $i_* = \min\{j \in \{1,2,\ldots \} \mid b_j \neq b'_j\}$
   for $\enc^n(x) =b_1b_2\cdots b_n$ and  $\enc^n(x^{\prime})=b'_1b'_2\cdots b'_n$
   unless  $\enc^n(x) = \enc^n(x')$. 
\end{proposition}

\begin{proposition}\label{prop:-heikai}
The $n$-th iterated tent code is right continuous, i.e., $\enc_{\mu}^n(x) = \enc_{\mu}^n(x+0)$. 
\end{proposition}

 These two facts make the arguments simple.  
 The following technical lemma is useful 
  to prove Propositions~\ref{prop:order} and \ref{prop:-heikai}, 
   as well as the arguments in Sections~\ref{sec:algo} and \ref{sec:average-complexity}. 
\begin{lemma}\label{lem:encode1}
  Let $x,x' \in [0,1)$ satisfy $x < x'$. 
  Let $\enc(x) =b_1b_2\cdots $ and  $\enc(x^{\prime})=b'_1b'_2\cdots $. 
 If $b_i = b'_i$ hold for all $i=1,\ldots,n$ then 
  \begin{align}
  \begin{cases}
    x_n < x^{\prime}_n  & \mbox{if $b_n = b^{\prime}_n = 0$, }\\ 
    x_n > x^{\prime}_n & \mbox{if $b_n = b^{\prime}_n = 1$}
 \end{cases}
  \label{eq:tentexpansion_lexicography_1}
  \end{align}
  holds. 
\end{lemma}

 Let ${\cal L}_{n,\mu}$ (or simply ${\cal L}_n$) denote the set of all $n$-bits tent codes, 
 i.e., 
\begin{align}
{\cal L}_n = \left\{ \enc_{\mu}^{n}(x) \in \{0,1\}^n \ \middle|\ x \in [0,1) \right\}
\label{def:tent-lang}
\end{align} 
 and we call ${\cal L}_n$ {\em tent language} (by $\mu \in (1,2)$). 
 Note that ${\cal L}_n \subsetneq \{0,1\}^n$ for $\mu \in (1,2)$. 

 Let ${\cal D}_{n,\mu}$ (or simply ${\cal D}_n$) denote a probability distribution over ${\cal L}_n$ 
  which follows $\enc^{n}(X) $ for $X$ is uniformly distributed over $[0,1)$, 
 i.e., ${\cal D}_n$ represents the probability of appearing $\bitseq_n \in {\cal L}_n$ 
   as given the initial condition $x$ uniformly at random. 
Our goal is to output ${\bf B} \in {\cal L}_n$ according to ${\cal D}_n$.  
To be precise,  this paper establishes the following theorem. 
\begin{theorem}\label{thm:main}
 Let $\mu \in (1,2)$ be a rational given by an irreducible fraction $\mu=c/d$. 
 Then, it is  possible to generate $B \in {\cal L}_n$ according to ${\cal D}_n$ in $\Order(\lg^2 n \lg^3 d / \lg^4 \mu)$ space in expectation 
 (as well as, with high probability). 
\end{theorem}
To prove it, 
 we will give an algorithm to generate ${\bf B} \sim {\cal D}_n$ in Section~\ref{sec:algo}, 
 where 
  ${\bf B} \sim {\cal D}_n$ denotes that a random variable ${\bf B} \in \{0,1\}^n$ 
    follows the distribution ${\cal D}_n$ (thus, ${\bf B} \in {\cal L}_n$ (at least probability $1$)). 
 Then, we will analyze the average space complexity in Section~\ref{sec:average-complexity}.

\section{Algorithm for the Proof of Theorem of \ref{thm:main}}\label{sec:algo}
 The goal of this section is to design a space efficient algorithm 
    to generate ${\bf B} \sim {\cal D}_n$ (recall Theorem~\ref{thm:main}), 
  and it will be presented as Algorithm~\ref{alg1} in Section~\ref{sec:algo-summary}. 
 To design a space efficient algorithm, 
   we need a compact representation for a recognition of a tent code. 
 For the purpose, 
   we have to develop notions and arguments, far from trivial, 
   that is namely 
   the notion of segment-type (in Section~\ref{ss:sectiontype}), 
   transition diagram of segment-types (in Section~\ref{sec:recog}), and 
   transition probability to realize ${\cal D}_n$ (in Section~\ref{sec:MC}).

\subsection{Compressing lemma---an intuition as an introduction}
\begin{figure}[tbp]
  \begin{tabular}{cc}
    \begin{minipage}[t]{.5\hsize}
      \centering
      \includegraphics[width=1\linewidth]{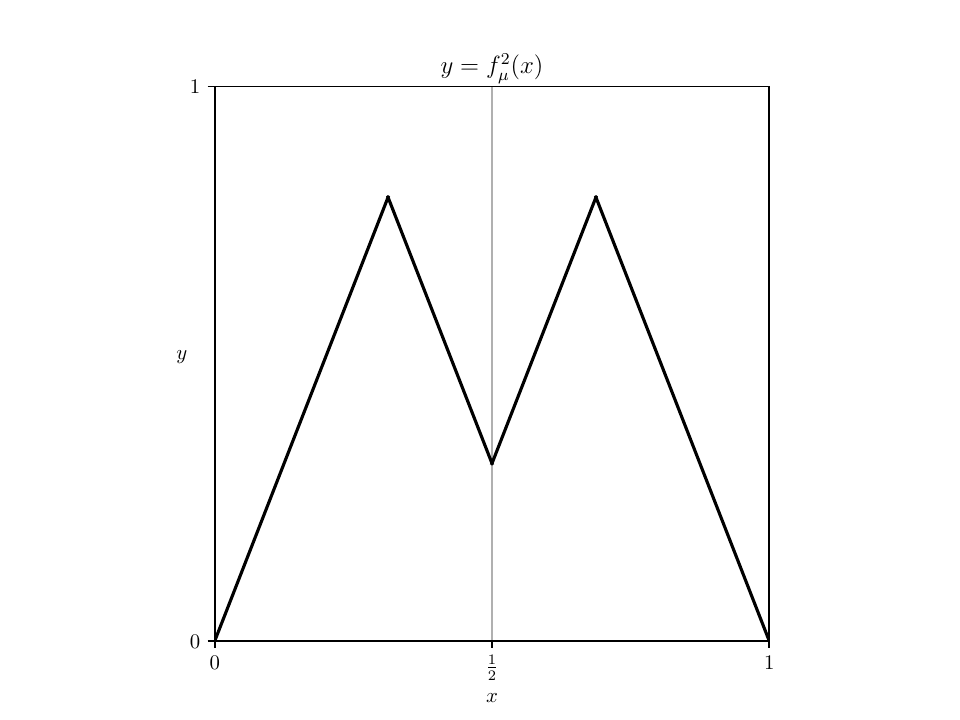}
      \subcaption{The $2$nd iterated tent map $\tent^2(x)$}
      \label{fig:compress_2}
    \end{minipage} &
    \begin{minipage}[t]{.5\hsize}
      \centering
      \includegraphics[width=1\linewidth]{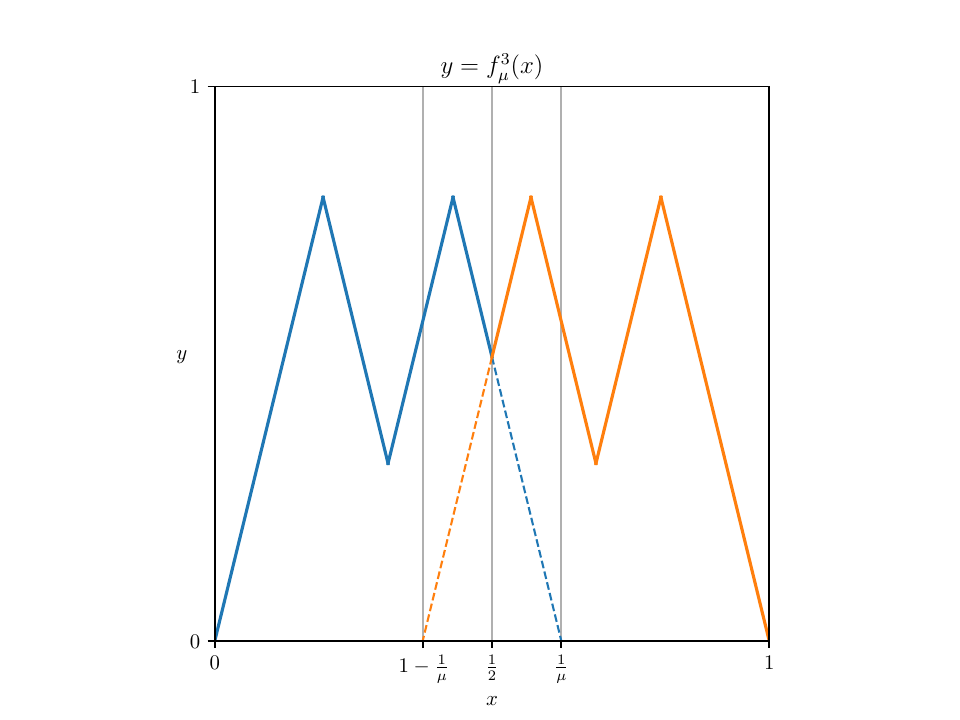}
      \subcaption{The $3$rd iterated tent map $\tent^3(x)$}
      \label{fig:compress_3}
    \end{minipage}
  \end{tabular}
  \caption{Lemma~\ref{lem:compress} implies that 
    $\tent^3$ consists of two $\tent^2$, 
     compressed in $1/\mu$ in $x$-axis direction and cut off at $1/2$.}
  \label{fig:compress}
\end{figure}

 As Figure~\ref{fig:expansion} shows, 
  the $n$-th iterated tent map looks complicated, consisting of exponentially many line segments, 
  and it definitely causes the ``sensitivity to initial conditions''  of a tent map.  
 Roughly speaking, 
  we will prove that 
   the line-segments of the $n$-th iterated map are classified into at most $2n$ classes 
    by the range ($y$ values) of line-segments 
 (see  Theorem~\ref{theo:number_of_type} in Section~\ref{sec:number_of_type}, for detail).  
 Before the precise argument, 
  we briefly remark a key observation for an intuition of the argument.

\begin{lemma}[Compressing lemma]\label{lem:compress}
Let 
\begin{align}
 \tilde{\tent}(x) = \begin{cases}
  \tent(x) &: x \leq \tfrac{1}{2}, \\
  1- \tent(x) &: x \geq \tfrac{1}{2}
 \end{cases}
 \label{def:tilde}
\end{align}
then, 
\begin{align}
 \tent^{n+1}(x) 
  &= \tent^n(\tilde{\tent}(x))
  =\begin{cases}
      \tent^{n}(\mu x) &: x \leq \frac{1}{2}, \\
      \tent^{n}(1- \mu (1 - x)) &: x \ge \frac{1}{2}
  \end{cases} 
  \label{eq:pres-tent}
 \end{align}
 holds for $x\in [0,1)$ and $n=1,2,\ldots$. 
\end{lemma}
\begin{proof}
 It is trivial for $x \leq 1/2$ since $\tilde{\tent}(x) = \tent(x)$. 
 Suppose $x \geq 1/2$. 
 Then, 
\begin{align*}
  \tent(\tilde{\tent}(x)) 
   &= \tent(1-\tent(x)) \\
   &= \tent(\tent(x))  && (\mbox{since $\tent(x) =\tent(1-x)$})\\
   &= \tent^2(x) 
\end{align*}
 holds, and we obtain the claim for $n=1$. 
 Once we obtain $\tent(\tilde{\tent}(x)) = \tent^2(x)$, 
  the claim is trivial for $n \geq 2$. 
\end{proof}
 We remark for \eqref{eq:pres-tent} 
  that let $x=1-t$ ($t \leq 1/2$) then $\tilde{\tent}(x) = 1- \mu (1 - x) = 1 -\mu(1-(1-t)) = 1- \mu t$.
 Lemma~\ref{lem:compress} gives an intuition that 
  $\tent^{n+1}$ consists of two $\tent^n$   
  compressed in $1/\mu$ in $x$-axis direction and cut off at $x=1/2$ (see Figure~\ref{fig:compress}).  
 This is a key observation that the line-segments are effectively classified by their projections on $y$-axis. 
 A precise argument follows.

\subsection{Sections and segment-types}\label{ss:sectiontype}
\subsubsection{Sections}
 To begin with, 
   we introduce a natural equivalent class over  $[0,1)$ with respect to $\enc^n$. 
 Let 
\begin{align}
  \se_n(\bitseq) \defeq \{x \in [0,1) \mid \enc^{n}(x) = \bitseq \} 
  \label{eq:section}
\end{align}
  for a bit sequence $\bitseq = b_{1} \cdots b_{n} \in {\cal L}_n$, and 
  we call it {\em section} (of $\bitseq$). 
  For convenience, 
   we define $\se(\bitseq_{n}) = \emptyset$ if $\bitseq_{n} \not\in {\cal L}_n$. 
 We also abuse $ \se_n(x)$ for $x \in [0,1)$ as $\se_n(\enc^n(x))$, and call it section of $x$.  
Let 
\begin{align}
{\cal S}_n = \{ \se_n(\bitseq) \mid \bitseq \in {\cal L}_n \} 
 = \{ \se_n(x) \mid x \in [0,1) \} 
\end{align}
 denote the whole set of sections provided by the $n$-th iterated tent map. 
 Since $\enc^n(x)$ is uniquely defined for any $x \in [0,1)$, 
   it is not difficult to see that the set of sections is a partition of $[0,1)$, 
 i.e., 
 $\bigcup_{S \in {\cal S}_n} \se =[0,1) $, and 
  $S \cap S'=\emptyset$ for $S \neq S$.

Due to the right continuity of $\enc^n$ (Proposition~\ref{prop:-heikai}), 
 every section is left-closed and right-open interval, i.e., 
\begin{align} 
 S=\left[\inf S,\sup S\right) 
 \label{eq:sec-heikai}
\end{align}
for any $S \in {\cal S}_n$. 
We  observe that every section is subdivided by an iteration of the tent map, 
 meaning that 
\begin{align}
 S_n(\bitseq) = S_{n+1}(\bitseq_n0) \cup S_{n+1}(\bitseq_n1)
 \label{eq:subdiv-sec}
\end{align}
holds for any  $\bitseq = b_1\cdots b_n \in {\cal L}_n$. 
As a consequence, we also observe that 
\begin{align}
 \bigcup_{\bitseq_{n} \in {\cal L}_n \mid b_1=0} \se(\bitseq_{n}) =[0,1/2) \qquad \mbox{and}  \qquad
 \bigcup_{\bitseq_{n} \in {\cal L}_n \mid b_1=1} \se(\bitseq_{n}) =[1/2,1) 
\label{eq:first-bit}
\end{align}
hold. 

\subsubsection{Segment-types}
Clearly, the size of ${\cal S}_n$ is as large as exponential to $n$. 
Interestingly, 
 we will prove that 
 ${\cal S}_n$ is classified by their {\em images} 
  into at most $2n$ types. 
 Let
\begin{align}
  \typefn(S) &\defeq \{ \tent^{n}(x) \mid x \in S \} 
  \label{def:type}
\end{align}
 for $S \in {\cal S}_n$, where 
 we call $\typefn(S)$ the {\em segment-type} of $S$. 
 We can observe every segment-type is a connected interval on $[0,1)$ since $\tent^n$ is a piecewise linear function.
 We abuse 
$ \typefn(\bitseq)$ for $\bitseq \in {\cal L}_n$ as  
  $\typefn(\se_n(\bitseq)) $. 
 For convenience, we also use 
$  \typefn^n(x)$ for $x \in[0,1)$ as $\typefn(\enc^n(x))$ ($=\typefn(\se_n(\enc^n(x)))$). 
Let  
\begin{equation}
  \typeset_n = \{\typefn(\bitseq) \subseteq [0,1)  \mid \bitseq \in {\cal L}_n \}
\end{equation}
denote the set of segment-types (of ${\cal L}_n$). 
For instance, 
\begin{align}
 \typeset_1 
  &= \{ \typefn(0), \typefn(1)\}  \nonumber\\
  &= \{ \tent(\se_1(0)), \tent(\se_1(1))\} \nonumber\\
  &= \{ \tent([0,\tfrac{1}{2})), \tent([\tfrac{1}{2},1))\} \nonumber\\
  &= \{ [0,\tfrac{\mu}{2}), (0,\tfrac{\mu}{2}]\} 
\label{eq:typeset1}
\end{align} 
holds for $n=1$. 
Then, Theorem~\ref{theo:number_of_type}, appearing later, claims $|\typeset_n| \leq 2n$. 
Before stating the theorem, 
 let us briefly illustrate the segment-types.

\subsubsection{Brief remarks on segment-types to get it}\label{sec:fig-type}
Figure~\ref{fig:expansion}, illustrating the iterated tent map $\tent^n$, shows many ``line-segments'' 
  provided by 
\begin{align}
 \left\{ (x,\tent^n(x)) \in \mathbb{R}^2 \mid x \in \se_n(\bitseq_{n}) \right\}
\label{eq:line-segment}
\end{align} 
 for each $\bitseq_{n} \in {\cal L}_n$. 
 It seems that 
   the line-segments corresponds to ${\cal S}_n$ (and hence ${\cal L}_n$) one-to-one, by the figure.  
 At this moment, we see from  Lemma~\ref{lem:encode1} that 
   $\tent^n$ is monotone increasing on $\se_n(\bitseq_n)$ if $b_n=0$, 
  otherwise, i.e., $b_n=1$, the map $\tent^n$ is monotone decreasing. 
 In fact, Theorem~\ref{theo:number_of_type} appearing later implies 
  the one-to-one correspondence between line-segments and ${\cal L}_n$. 
  
 Anyway, a segment-type $\typefn(\se_n(\bitseq_n))$ provides another representation of a section $\se_n(\bitseq_n)$, 
  as well as the corresponding bit sequence $\bitseq_n \in {\cal L}_n$, 
   through the line-segment. 
 Interestingly, we observe the following connection between segment-types and the tent code, 
   from Lemma~\ref{lem:encode1}  and Proposition~\ref{prop:-heikai}. 
\begin{lemma}\label{lem:heikai_n}
 Let $\bitseq_n \in {\cal L}_n$ and let $\typei = \typefn(\bitseq_n)$. 
 Let $v=\inf \typei$ and let $u=\sup \typei$, for convenience.  
 Then, 
 \begin{align}
 \typei = \begin{cases}
 [v,u) & \mbox{if $b_n=0$ (i.e., the line-segment is monotone increasing),} \\
 (v,u] & \mbox{if $b_n=1$ (i.e., the line-segment is monotone decreasing)}
 \end{cases}
 \end{align}
  holds. 
\end{lemma}
\begin{proof}
 We recall that every section is left-closed and right-open, i.e., $S=[\inf S, \sup S)$ by Proposition~\ref{prop:-heikai}. 
 We also remark that $\tent^n$ is continuous since $\tent$ is continuous by \eqref{eq:tentmap}. 

 Consider the case $b_n=0$. 
 As we stated above, Lemma~\ref{lem:encode1} implies that 
  $\tent^n$ is monotone increasing on $S=\se_n(\bitseq_n)$ in the case. 
 Thus, $v = \inf \tent^n(S) = \tent^n(\inf S)$ and $u = \sup \tent^n(S) = \tent^n(\sup S)$. 
 Since $S=[\inf S, \sup S)$, we see $v \in \typefn(S)$ and $u \not\in \typefn(S)$. 
 We obtain $\typei = [v,u)$ in the case. 

 The case $b_n=1$ is similar. 
 By Lemma~\ref{lem:encode1}, 
  $\tent^n$ is monotone decreasing on $S=\se_n(\bitseq_n)$ in the case. 
 Thus, $v = \inf \tent^n(S) = \tent^n(\sup S)$ and $u = \sup \tent^n(S) = \tent^n(\inf S)$. 
 Since $S=[\inf S, \sup S)$, we see $v \not\in \typefn(S)$ and $u \in \typefn(S)$. 
 We obtain $\typei = (v,u]$ in the case. 
\end{proof}

\subsection{The number of segment-types}\label{sec:number_of_type}
 Concerning the set of segment-types 
  we will establish the following theorem, which claims  $|\typeset_n| \leq 2n$. 
\begin{theorem}
  \label{theo:number_of_type}
 Let $\mu \in (1, 2)$. 
 Let $\bitseqc_i = \enc^i(\frac{1}{2})$, and let  
 \begin{align}
 \typei_i = \typefn(\bitseqc_i) \hspace{2em}\mbox{and}\hspace{2em}
 \flip{\typei}_i = \typefn(\flip{\bitseqc}_i)
 \label{def:typei}
 \end{align}
  for $i=1,2,\ldots$. 
Then, 
  \begin{align}
    \typeset_n  = \bigcup_{i=1}^{n_*} \left\{ \typei_i, \flip{\typei}_i \right\}
    \label{eq:number_of_type}
  \end{align}
  for $n \geq 1$, 
 where 
  $n_* = \min (\{i \in \{1,2,\ldots,n-1\} \mid \typei_{i+1} \in \typeset_i\} \cup \{n\})$. 
\end{theorem}

\subsubsection{Weak lemma}
 We here are concerned with the following weaker version, which is enough to claim $|\typeset_n| \leq 2n$, 
   to follow the main idea of the arguments in this section. 
\begin{lemma}\label{lem:type_subset}
 Let $\bitseqd_i = \min\{ \bitseq_i \in {\cal L}_i \mid b_1 = 1  \mbox{ where } \bitseq_i  = b_1 \cdots b_i \}$, and 
 let $\bitseqd'_i = \max\{ \bitseq_i \in {\cal L}_i \mid b_1 = 0  \mbox{ where } \bitseq_i  = b_1 \cdots b_i \}$ 
  where $\min$ and $\max$ are according to the lexicographic order.  
 Then, 
  \begin{align}
    \typeset_n  \subseteq \bigcup_{i=1}^n \left\{ \typefn(\bitseqd_i), \typefn(\bitseqd'_i) \right\}
    \label{eq:type_subset}
  \end{align}
 for $n \geq 1$. Furthermore, $\typeset_{n-1} \setminus \typeset_n  = \emptyset$. 
\end{lemma}
 To get Theorem~\ref{theo:number_of_type} from Lemma~\ref{lem:type_subset}, 
   we further need to prove $\bitseqd_i = \bitseqc_i$ and $\bitseqd'_i  = \flip{\bitseqc_i}$ (Lemma~\ref{lem:c=d}), and 
   the existence of $n_*$ (Lemma~\ref{lem:number_of_type3}), 
    which are not difficult but require somehow bothering arguments to prove. 

 Lemma~\ref{lem:type_subset} may intuitively and essentially follow from Lemma~\ref{lem:compress}, 
   considering the representation of a line-segment of $\tent^n$ by a segment-type  (see Section~\ref{sec:fig-type}). 
 For a proof of Lemma~\ref{lem:type_subset}, 
  we start by giving the following strange recursive relationship between sections (cf. Lemma~\ref{lem:compress}, for an intuition). 

\begin{lemma}\label{lem:section-shift}
  $\tilde{\tent}(S_n(\bitseq_n)) \subseteq S_{n-1}(b_2 \cdots b_n)$ holds for any $\bitseq_n = b_1 b_2 \cdots b_n \in {\cal L}_n$. 
Furthermore, $\tilde{\tent}(S_n(\bitseq_n)) \neq S_{n-1}(b_2 \cdots b_n)$ 
 only when $\bitseq_n = \bitseqd_n$ or $\bitseqd'_n$. 
 \end{lemma}
The former part of  Lemma~\ref{lem:section-shift} comes from the following fact. 
\begin{lemma}\label{lem:obs-rephrase}
Let $x \in [0,1)$ and let $\tilde{x} = \tilde{\tent}(x)$.  
Suppose
 $\enc^n(x) = b_1\cdots b_n$, and  
 $\enc^{n-1}(\tilde{x}) = d_1\cdots d_{n-1}$. 
Then, $b_{i+1} = d_i$ for $i = 1,2,\ldots,n-1$. 
\end{lemma}
\begin{proof}
Let $x_{i+1}=\tent^{i+1}(x)$ and $\tilde{x}_i=\tent^i(\tilde{x}) =\tent^i(\tilde{f}(x))$,  
then $x_{i+1} = \tilde{x}_i$ holds for $i=1,2,\ldots$ by Lemma~\ref{lem:compress}. 
Recall \eqref{def:encode2} that $b_{i+1}$ depends only on $x_i$ and $b_i$, as well as $d_i$ depends on $\tilde{x}_{i-1}$ and $d_{i-1}$, for $i=2,3,\ldots$. 
Thus, it is enough to prove $b_2 = d_1$, 
 then we inductively obtain the claim since $x_{i+1} = \tilde{x}_i$. 

Firstly,  we consider the case of $x<1/2$. 
Notice that $x_1 = \tilde{x} = \tent(x)$, in the case. 
We consider two cases whether $\tent(x) < 1/2$ or not. 
 Suppose $\tent(x) < 1/2$. 
 Then, $b_2 =0$ by \eqref{def:encode2}, 
  as well as $d_1 = 0$ by \eqref{def:encode0}, and hence $b_2 =d_1$. 
 Suppose $\tent(x) \geq 1/2$. 
  Then, $b_2 =1$ by \eqref{def:encode2}, 
  as well as $d_1 = 1$ by \eqref{def:encode0}, and hence $b_2 =d_1$. 
  
Next, we consider the case $x \geq 1/2$. 
In the case, 
  $x_1 = \tent(x)$ and $\tilde{x} = 1-\tent(x)$. Notice that $x_1 = 1-\tilde{x}$ holds.  
We consider two cases whether $x_1 > 1/2$ or not. 
 Suppose $x_1 > 1/2$. 
 Then, $b_2 =0$ by \eqref{def:encode2}  
 while $d_1 = 0$ by \eqref{def:encode0}  since $\tilde{x} = 1-x_1 <1/2$. We obtain $b_2 =d_1$. 
 Suppose $x_1 \leq 1/2$.  
  Then, $b_2 =1$ by \eqref{def:encode2}, 
  while $d_1 = 1$ by \eqref{def:encode0}, accordingly $b_2 =d_1$. 
 We obtain the claim. 
\end{proof}

\begin{proof}[Proof of Lemma~\ref{lem:section-shift}]
 The former part is almost trivial from Lemma~\ref{lem:obs-rephrase}:  
 In fact, if $x \in \se_n(\bitseq_n)$ then $\enc^n(x) = \bitseq_n$. 
 Lemma~\ref{lem:obs-rephrase} implies $\enc^{n-1}(\tilde{\tent}(x)) = b_2 \cdots b_n$
  meaning that $\tilde{\tent}(x) \in \se_{n-1}(b_2 \cdots b_n)$. 

 Now, we prove, for the latter part, that  
\begin{align}
  \tilde{\tent}(\se_n(\bitseq_n)) \supseteq \se_{n-1}(b_2 \cdots b_n) 
\end{align}
  unless $\bitseq_n = \bitseqd_n$ or $\bitseqd'_n$. 
 Firstly, we consider the case $b_1 =0$. 
 Suppose $\se_{n-1}(b_2 \cdots b_n) = [u,v)$ where $u<v$. 
 Let $y \in \se_{n-1}(b_2 \cdots b_n)$, then 
  we claim that $y \in \tilde{\tent}(\se_n(\bitseq_n))$ 
 if $v \leq \frac{\mu}{2}$. 
 For the purpose, 
   we calculate $\enc^n(\frac{y}{\mu})$.  
 Clearly, $b_1=0$ since $\frac{y}{\mu} < \frac{v}{\mu} \leq \frac{1}{2}$. 
 By lemma~\ref{lem:obs-rephrase}, 
 $b_2\cdots b_n = \enc^{n-1}(\tilde{\tent}(\frac{y}{\mu})) = \enc^{n-1}(y) = b_2 \cdots b_n$ since $y \in \se_{n-1}(b_2 \cdots b_n)$. 
 Thus, we obtain $\enc^n(\frac{y}{\mu}) = \bitseq_n$, meaning that $\frac{y}{\mu} \in \se_n(\bitseq_n)$. 
 Now it is easy to see $y \in \tilde{\tent}(\se_n(\bitseq_n))$ since $\tilde{\tent}(\frac{y}{\mu}) = y$.  
 Recall that the set of sections ${\cal S}_n$ is a partition of $[0,1)$ for each $n$, 
  that the sections are allocated in [0,1) in lexicographic order by Proposition~\ref{prop:order}, and 
  that $\tilde{\tent}(x) = \mu x$ for $x<\frac{1}{2}$ is clearly order preserving.  
 Thus, only $\se_{n-1}(d'_2 \cdots d'_n) = [u,v)$ may violate $v<\frac{\mu}{2}$, 
   where $\bitseqd'_n = \max\{ d'_1 \cdots d'_n \in {\cal L}_n \mid d'_1 = 0 \}$. 
  
 The case of $b_1=1$ is similar.  
 Suppose $\se_{n-1}(b_2 \cdots b_n) = [u,v)$ where $u<v$. 
 Let $y \in \se_{n-1}(b_2 \cdots b_n)$, then 
  we claim that $y \in \tilde{\tent}(\se_n(\bitseq_n))$ 
 if $u \geq 1- \frac{\mu}{2}$. 
 For the purpose, 
   we calculate $\enc^n(1-\frac{1-y}{\mu})$, 
   where we remark that $1-\frac{1-y}{\mu} \geq 1-\frac{1-u}{\mu} \geq \frac{1}{2}$ and 
   $\tilde{\tent}(1-\frac{1-y}{\mu}) = 1-\tent(1-\frac{1-y}{\mu}) = 1-\mu(1-(1-\frac{1-y}{\mu})) = y$.  
 Clearly, $b_1=1$ since $1-\frac{1-y}{\mu} \geq \frac{1}{2}$. 
 By lemma~\ref{lem:obs-rephrase}, 
 $b_2\cdots b_n = \enc^{n-1}(\tilde{\tent}(1-\frac{1-y}{\mu})) = \enc^{n-1}(y) = b_2 \cdots b_n$ since $y \in \se_{n-1}(b_2 \cdots b_n)$. 
 Thus, we obtain $\enc^n(1-\frac{1-y}{\mu}) = \bitseq_n$, meaning that $1-\frac{1-y}{\mu} \in \se_n(\bitseq_n)$. 
 Now it is easy to see $y \in \tilde{\tent}(\se_n(\bitseq_n))$ since $\tilde{\tent}(1-\frac{1-y}{\mu}) = y$.  
 Recall that the set of sections ${\cal S}_n$ is a partition of $[0,1)$ for each $n$, 
  that the sections are allocated in [0,1) in lexicographic order by Proposition~\ref{prop:order}, and 
  that $\tilde{\tent}(x) = 1-\tent(x) = 1-\mu(1-x) = \mu x + (1-\mu)$ for $x \geq 1/2$ is clearly order preserving.  
 Thus, only $\se_{n-1}(d_2 \cdots d_n) = [u,v)$ may violate $u \geq 1-\frac{\mu}{2}$, 
   where $\bitseqd_n = \min\{ d_1 \cdots d_n \in {\cal L}_n \mid d_1 = 1 \}$. 
\end{proof}

Now, we are ready to prove Lemma~\ref{lem:type_subset}. 
\begin{proof}[Proof of Lemma~\ref{lem:type_subset}] 
 The claim is trivial for $n=1$ (recall \eqref{eq:typeset1}). 
 We prove 
  $\typeset_n \subseteq \typeset_{n-1} \cup \{\typefn(\bitseqd_n),\typefn(\bitseqd'_n)\}$ for $n \geq 2$. 
For the purpose, 
 we claim $\typefn(\bitseq_n) = \typefn(b_2 \cdots b_n)$ unless $\bitseq_n = \bitseqd_n$ or $\bitseqd'_n$.  
In fact, 
\begin{align*}
 \typefn(\bitseq_n) 
  &= \{\tent^n(x) \mid x \in \se_n(\bitseq_n)\} 
    && (\mbox{by definition~\eqref{def:type}}) \\
  &= \{\tent^{n-1}(\tilde{\tent}(x)) \mid x \in \se_n(\bitseq_n)\} 
    && (\mbox{by Lemma~\ref{lem:compress}}) \\
  &= \{\tent^{n-1}(y) \mid y \in \tilde{\tent}(\se_n(\bitseq_n))\} 
    && (\mbox{set $y = \tilde{\tent}(x)$}) \\
  &= \{\tent^{n-1}(y) \mid y \in \se_n(b_2\cdots b_n)\} 
    && (\mbox{by Lemma~\ref{lem:section-shift}, w/  $\bitseq_n \neq \bitseqd_n,\bitseqd'_n$}) \\
  &= \typefn(b_2 \cdots b_n) 
    && (\mbox{by definition~\eqref{def:type}}) 
\end{align*}
holds. 
 Clearly, $\typefn(b_2 \cdots b_n) \in \typeset_{n-1}$ by definition, and 
 we inductively obtain the claim. 
\end{proof}

\subsubsection{Lemmas for $\bitseqc_n = \bitseqd_n$ and  $\flip{\bitseqc_n} = \bitseqd'_n$}
 As a remaining part of the proof of Theorem~\ref{theo:number_of_type}, 
  here we just refer  to the following facts. See Section~\ref{apx:c=d} for a proof.
\begin{lemma}\label{lem:c=d}
 Let $\bitseqc_n = \enc^n(\frac{1}{2})$. 
 Then, $\bitseqc_n = \bitseqd_n$ and  $\flip{\bitseqc_n} = \bitseqd'_n$ hold, where 
  $\bitseqd_n = \min\{ \bitseq_n \in {\cal L}_n \mid b_1 = 1  \mbox{ where } \bitseq_n  = b_1 \cdots b_n \}$ and 
  $\bitseqd'_n = \max\{ \bitseq_n \in {\cal L}_n \mid b_1 = 0  \mbox{ where } \bitseq_n  = b_1 \cdots b_n \}$.  
\end{lemma}
The former claim is trivial by Proposition~\ref{prop:order} and \eqref{eq:first-bit}. 
For a proof of the latter claim, we use the following fact, which intuitively seems obvious from  Lemma~\ref{lem:compress}. 
\begin{lemma}\label{lem:hanten}
$\flip{\enc^n(x)} = \enc^n(1-x)$ 
 unless $x = \min S_n(x)$. 
\end{lemma}

 As a consequence of Lemma~\ref{lem:hanten} with Lemma~\ref{lem:heikai_n}, we remark the following convenient  observation. 
\begin{lemma}\label{lem:hanten-type}
If $\typefn(\bitseq_n) = [v,u)$ then $\typefn(\flip{\bitseq_n}) = (v,u]$, and vice versa. 
\end{lemma}
See Section~\ref{apx:c=d} for proofs of Lemmas~\ref{lem:c=d} and \ref{lem:hanten}, which are not difficult but a bit lengthy.

\subsection{Transitions over $\typeset_n$ and Recognition of ${\cal L}_n$}\label{sec:recog}
\subsubsection{Transitions between segment-types}
 Recall \eqref{def:encode2} that 
   $b_{i+1}$ is determined by $b_i$ and $\tent^i(x)$ 
   to compute $\enc^n(x) = b_1 \cdots b_n$ for $x \in [0,1)$.   
 More precisely, Lemma~\ref{lem:encode1} lets us know 
  if $b_i = 0$ then $\typefn(\bitseq_i) = [v,u)$, otherwise i.e., $b_i = 1$ then $\typefn(\bitseq_i) = (v,u]$. 
  Here, we establish the following lemma, 
  about the transitions over segment-types. 

\begin{lemma}[Transitions of segment-types]\label{lem:transition}
Let $x \in [0,1)$. 
 (1) Suppose $\typefn^n(x) = [v,u)$ ($v<u$). 
 We consider three cases concerning the position of $\tfrac{1}{2}$ relative to $[v,u)$. 
\vspace{-1ex}
\begin{itemize}\setlength{\itemindent}{4em}\setlength{\parskip}{0.5ex}\setlength{\itemsep}{0cm} 
\item[Case 1-1:]  $v < \frac{1}{2} < u$. 
\vspace{-1ex}
\begin{itemize}\setlength{\itemindent}{4em}\setlength{\parskip}{0.5ex}\setlength{\itemsep}{0cm} 
\item[Case 1-1-1.] If $\tent^n(x) < 1/2$ then $\typefn^{n+1}(x) = [\tent(v),\tent(\tfrac{1}{2}))$, and $b_{n+1}=0$. 
\item[Case 1-1-2.] If $\tent^n(x) \geq 1/2$ then $\typefn^{n+1}(x) =  (\tent(u),\tent(\tfrac{1}{2})]$, and $b_{n+1}=1$. 
\end{itemize}
\item[Case 1-2:] $u \leq \frac{1}{2}$. Then $\typefn^{n+1}(x) = [\tent(v),\tent(u))$, and $b_{n+1}=0$.  
\item[Case 1-3:] $v \geq \frac{1}{2}$. Then $\typefn^{n+1}(x) = (\tent(u),\tent(v)]$, and $b_{n+1}=1$. 
\end{itemize}
\vspace{-1ex}

\noindent (2) Similarly, suppose $\typefn^n(x) = (v,u]$ ($v<u$). 
\vspace{-1ex}
\begin{itemize}\setlength{\itemindent}{4em}\setlength{\parskip}{0.5ex}\setlength{\itemsep}{0cm} 
\item[Case 2-1:]  $v < \frac{1}{2} < u$. 
\vspace{-1ex}
\begin{itemize}\setlength{\itemindent}{4em}\setlength{\parskip}{0.5ex}\setlength{\itemsep}{0cm} 
\item[Case 2-1-1.] If $\tent^n(x) \leq 1/2$ then $\typefn^{n+1}(x) = (\tent(v),\tent(\tfrac{1}{2})]$, and $b_{n+1}=1$. 
\item[Case 2-1-2.] If $\tent^n(x) > 1/2$ then $\typefn^{n+1}(x) =  [\tent(u),\tent(\tfrac{1}{2}))$, and $b_{n+1}=0$. 
\end{itemize}
\item[Case 2-2:] $u \leq \frac{1}{2}$. Then $\typefn^{n+1}(x) = (\tent(v),\tent(u)]$, and $b_{n+1}=1$. 
\item[Case 2-3:] $v \geq \frac{1}{2}$. Then $\typefn^{n+1}(x) = [\tent(u),\tent(l))$, and $b_{n+1}=0$. 
\end{itemize}
\end{lemma} 
We remark that Lemma~\ref{lem:transition} is rephrased by 
 \begin{align} 
 \typefn^{n+1}(x) &= \begin{cases}\begin{cases}
 [\tent(v),\tent(\tfrac{1}{2}))  &: \tent^n(x) < 1/2 \\
 (\tent(u),\tent(\tfrac{1}{2})] &: \tent^n(x) \geq 1/2 
\end{cases} 
  &: v < \tfrac{1}{2} < u \\ 
 [\tent(v),\tent(u))  &:  u \leq \tfrac{1}{2}  \\
 (\tent(u),\tent(v)]  &:  v \geq \tfrac{1}{2} 
 \end{cases}
&&\mbox{if $\typefn^n(x) = [v,u)$, } 
\label{eq:transition1}\\
 \typefn^{n+1}(x) &= \begin{cases}\begin{cases}
 (\tent(v), \tent(\tfrac{1}{2})]  &: \tent^n(x) \leq 1/2 \\
 [\tent(u),\tent(\tfrac{1}{2})) &: \tent^n(x) > 1/2 
\end{cases} 
  &: v < \tfrac{1}{2} < u \\ 
 (\tent(v),\tent(u)]  &: u \leq \tfrac{1}{2}  \\
 [\tent(u),\tent(v))  &: l \geq \tfrac{1}{2}  
 \end{cases}
&&\mbox{if $\typefn^n(x) = (v,u]$.} 
\label{eq:transition2}
\end{align} 
\begin{proof}[Sketch of proof]
 The proof idea is similar to Lemma~\ref{lem:heikai_n}. 
 We here prove Case 1-1. Other cases are similar (see Section~\ref{apx:transition} for the complete proof).  

 To begin with, we recall three facts. 
 i) The segment-type is defined by $\typefn^n(x) = \{\tent^n(x) \mid x \in \se_n(\bitseq_n) \}$ by \eqref{def:type}. 
 ii) $\se_n(\bitseq_n) = \se_{n+1}(\bitseq_n0) \cup \se_{n+1}(\bitseq_n1)$ by \eqref{eq:subdiv-sec}. 
 Particularly,  $\inf\se_n(\bitseq_n) = \inf \se_{n+1}(\bitseq_n0)$ and $\sup \se_n(\bitseq_n) = \sup\se_{n+1}(\bitseq_n1)$ 
  by Proposition~\ref{prop:order}. 
 iii) $b_{n+1}$ depends on $b_n$ and $x_n$ by \eqref{def:encode2}. 

 Suppose $\typefn^n(x) = [v,u)$ ($v<u$), and $l < \frac{1}{2} < u$ (Case 1-1).  
 Notice that $\tent^n$ is monotone increasing in $\se_n(x)$ by Lemma~\ref{lem:heikai_n}. 
 Let $x^* \in \se_n(x)$ satisfy $\tent^n(x^*) (= x^*_n) = \frac{1}{2}$. 
 Clearly, $\se_n(\bitseq_n)$ is divided into $\se_{n+1}(\bitseq_n0)$ and $\se_{n+1}(\bitseq_n1)$ at $x^*$,  
 i.e., $\se_n(\bitseq_n)=[\inf \se_{n+1}(\bitseq_n0),x^*) \cup [x^*,\se_{n+1}(\bitseq_n1))$. 

  If $\tent^n(x) < 1/2$ then $x_n \in [l,\frac{1}{2})$. 
  Thus, $\tent^{n+1}(x) =\tent(x_n) = \mu x_n$. 
  Accordingly, 
  $\tent^{n+1}(\inf \se_{n+1}(x)) = \tent(\tent^n(\inf \se_n(x))) =\tent(l) = \mu l 
     \leq \tent(x_n)  
     < \mu \frac{1}{2} = \tent(\frac{1}{2}) =\tent(x^*_n)=\tent^{n+1}(x^*)$ since $l \leq x_n < \frac{1}{2}$ 
     (cf. the proof of  Lemma~\ref{lem:heikai_n}). 
   We obtain $\typefn^{n+1}(x) = [\tent(l),\tent(\tfrac{1}{2}))$.  $b_{n+1}=0$ is clear by Lemma~\ref{lem:heikai_n}. 

  If $\tent^n(x) \geq 1/2$ then $x_n \in [\frac{1}{2},u)$. 
  Thus, $\tent^{n+1}(x) =\tent(x_n) = \mu (1-x_n)$. 
  Accordingly, $\tent(\frac{1}{2}) = \mu (1-\frac{1}{2}) \geq \tent(x_n)  > \mu (1-u) = \tent(u)$ since $\frac{1}{2} \leq x_n < u$. 
  We obtain $\typefn^{n+1}(x) = (\tent(u),\tent(\tfrac{1}{2})]$. 
\end{proof}

 Eqs.~\eqref{eq:transition1} and \eqref{eq:transition2}, 
   which just rephrase Lemma~\ref{lem:transition}, 
   show a transition rule over segment-types, 
   where the next single bit of a tent code makes a transition, and vice versa; 
 Let $\typefn^n(x) = \typej$ and $\typefn^{n+1}(x) = J'$.  
 If $\typej$ satisfies the cases of 1-2, 1-3, 2-2 or 2-3, 
  then $b_{n+1}$ and $\typej'$ is uniquely determined, 
 while $\typej'$ and $b_{n+1}$ depends on each other in the cases of 1-1 or 2-1. 
For convenience, 
 let $\delta(\typej,b) = \typej'$ denote 
  that 
   if $\typefn^n(x) = \typej$ and $b_{n+1}=b$ then $\typefn^{n+1}(x) = \typej'$, 
  according to \eqref{eq:transition1} and \eqref{eq:transition2}. 
 Notice that 
  Lemma~\ref{lem:type_subset} implies that 
   $\typefn^i(x) = \typefn^j(x') = J$ may hold even for $i \neq j$;  
   for a typical instance,  
   $\typefn(\bitseq_n) = \typefn(b_2 \cdots b_n)$ holds 
    as used in the proof of Lemma~\ref{lem:type_subset}. 
  Then, we observe the following fact. 
\begin{lemma}\label{lem:number_of_type3}
 Suppose $\typefn(\bitseqc_{n+1}) \in \typeset_n$,  
   where $\bitseqc_{n+1} = \enc^{n+1}(\tfrac{1}{2})$. 
 Then, $\typefn(\flip{\bitseqc_{n+1}}) \in \typeset_n$, too. 
 Furthermore, 
$\typeset_{n'} = \typeset_n$ for any $n'$ satisfying $n' \geq n$. 
\end{lemma}
 Lemma~\ref{lem:number_of_type3} 
   might be intuitively obvious from Lemmas~\ref{lem:compress} and \ref{lem:type_subset}. 
 We prove it using Lemma~\ref{lem:transition}. 
\begin{proof}
 The former claim is easy from Lemma~\ref{lem:hanten}. 
 Suppose $\typefn(\bitseqc_{n+1}) = \typefn(\bitseq_k)$ for a $\bitseq_k \in {\cal L}_k$ ($k \leq n$). 
 Then, 
   Lemma~\ref{lem:hanten} implies that $\typefn(\flip{\bitseq_k}) \in {\cal L}_k$, and then
    $\typefn(\flip{\bitseqc_{n+1}}) = \typefn(\flip{\bitseq_k})$ holds by Lemma~\ref{lem:compress}. 
 
 We prove the latter claim using Lemma~\ref{lem:transition}. 
 By Lemmas~\ref{lem:type_subset} and \ref{lem:c=d}, 
   we know any $\typefn(\bitseq_{n+1}) \in \typeset_n$ unless  $\bitseq_{n+1} = \bitseqc_{n+1}$ or $\flip{\bitseqc_{n+1}}$. 
 By the hypothesis and the former claim, $\typefn(\bitseqc_{n+1})$ and $\typefn(\flip{\bitseqc_{n+1}}) \in \typeset_n$.  
 Thus, we obtain $\typefn(\bitseq_{n+1}) \in \typeset_n$ for any $\bitseq_{n+1} \in {\cal L}_{n+1}$. 
 In other words, 
   $\delta(J,b)  \in \typeset_n$ holds for any $J \in \typeset_n$ by Lemma~\ref{lem:transition}. 
 Inductively, we obtain $\typeset_{n'} \subseteq \typeset_n$ for any $n'$ satisfying $n' \geq n$. 
 The proof of Lemmas~\ref{lem:type_subset} implies that $\typeset_{n'} \supseteq \typeset_n$. 
\end{proof}
 Now, Theorem~\ref{theo:number_of_type} is clear by Lemmas~\ref{lem:type_subset}, \ref{lem:c=d} and \ref{lem:number_of_type3}. 

\begin{figure}[tb]
  \centering
  \begin{tikzpicture}
    [node distance=.5cm, every loop/.style={looseness=2}]
    \node [state, initial left] (v0) [] {$q_{0}$};
    \node [state] (a1) [above right = of v0] {$\typei_{1}$};
    \node [state] (a2) [right = of a1] {$\typei_{2}$};
    \node [state] (a3) [right = of a2] {$\typei_{3}$};
    \node [state, draw=none] (a4) [right = of a3] {$\cdots$};
    \node [state] (a5) [right = of a4] {$\typei_{n}$};
    \node [state] (b1) [below right = of v0] {$\flip{\typei}_{1}$};
    \node [state] (b2) [right = of b1] {$\flip{\typei}_{2}$};
    \node [state] (b3) [right = of b2] {$\flip{\typei}_{3}$};
    \node [state, draw=none] (b4) [right = of b3] {$\cdots$};
    \node [state] (b5) [right = of b4] {$\flip{\typei}_{n}$};
    \path [->] (v0) edge [] node [above left] {1} (a1);
    \path [->] (a1) edge [] node [above] {0} (a2);
    \path [->] (a2) edge [] node [above] {0} (a3);
    \path [->] (a3) edge [] node [above] {} (a4);
    \path [->] (a4) edge [] node [above] {} (a5);
    \path [->] (a1) edge [loop above] node [above] {1} ();
    \path [->] (a2) edge [bend left] node [pos=.2, right] {1} (b2);
    \path [->] (v0) edge [] node [below left] {0} (b1);
    \path [->] (b1) edge [] node [below] {1} (b2);
    \path [->] (b2) edge [] node [below] {1} (b3);
    \path [->] (b3) edge [] node [below] {} (b4);
    \path [->] (b4) edge [] node [below] {} (b5);
    \path [->] (b1) edge [loop below] node [below] {0} ();
    \path [->] (b2) edge [bend left] node [pos=.2, left] {0} (a2);
  \end{tikzpicture}
  \caption{Transition diagram for ${\cal L}_n$. }
  \label{fig:stg1}
\end{figure}
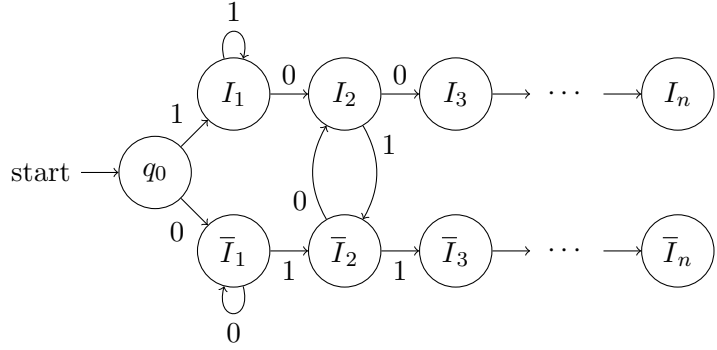

\subsubsection{Recognition of a tent language}
 Lemma~\ref{lem:transition} provides a natural finite state machine\footnote{
   Precisely, we need a ``counter'' for the length $n$ of the string, 
    while notice that our main goal is not to design an automaton for ${\cal L}_n$. 
   Our main target Theorem~\ref{thm:main} assumes a probabilistic Turing machine, where 
   obviously  we can count the length $n$ of a sequence in $\Order(\log n)$ space.  
   } to recognize/generate ${\cal L}_n$. 
 We define the set of states by 
  $Q_n = \{q_0\} \cup \{\emptyset\} \cup \typeset_n$, 
     where $q_0$ is the initial state, and $\emptyset$ denotes the unique reject state. 
   Recall  $\typeset_n = \bigcup_{i=1}^{n_*}\{ \typei_i, \flip{\typei}_i\}$ 
   where $\typei_i = \enc^i(\tfrac{1}{2})$ and $\flip{\typei}_i = \flip{\enc^i(\tfrac{1}{2})}$ by Theorem~\ref{theo:number_of_type}. 
 We let $\delta\colon Q_{n-1} \times \{0,1\} \to Q_n$ denote the state transition function, 
   which is defined  by 
\begin{align}
 \delta(\typej,b) = \typej'
\label{def:delta}
\end{align}
 according to \eqref{eq:transition1} and \eqref{eq:transition2} as far as $J$ and $b$ are consistent. 
 Let $\delta(q_0,1) = \typei_1$ and $\delta(q_0,0) = \flip{\typei}_1$. 
 For convenience, we define $\delta(\typej,b)=\emptyset$ 
  if the pair $\typej$ and $b$ contradicts to \eqref{eq:transition2}, 
  precisely 
\begin{align}\begin{cases}
 \mbox{$\typej = (v,u]$ and $v \geq \frac{1}{2}$} &\mbox{(cf. Case 1-3)}   \\
 \mbox{$\typej = [v,u)$ and $u \leq \frac{1}{2}$} &\mbox{(cf. Case 2-2)}   \\
 \mbox{$\typej = [v,u)$ and $u \leq \frac{1}{2}$} &\mbox{(cf. Case 1-2)}   \\
 \mbox{$\typej = (v,u]$ and $v \geq \frac{1}{2}$} &\mbox{(cf. Case 2-3)}   
\end{cases}
\end{align} 
are the cases, where $v=\inf \typej$ and  $u=\sup \typej$. 

 Now it is not difficult to see from Lemma~\ref{lem:transition} that 
   we can trace a path 
   starting from $q_0$ 
   according to $\delta$ 
   provided by $\bitseq_n = b_1 \cdots b_n \in \{0,1\}^n$ 
   if, and only if, $\bitseq \in {\cal L}_n$.
 For the latter argument, 
   we let $(Q_n,\delta)$ denote the state transition diagram (directed graph with labeled arcs; see Figure~\ref{fig:stg1}), where 
    $Q_n$ denote the set of states (vertices) and 
    $\delta\colon Q_{n-1}\to Q_n$ denote the set of transitions (arcs) labeled by $\{0,1\}$. 
 For convenience, 
  let $\outn(\typej) = \{\delta(\typej,0), \delta(\typej,1)\}\setminus\{\emptyset\}$ denotes the outneighbers of $\typej$ on the diagram $(Q_n,\delta)$.  
 Then, we note that 
\begin{align}
 |\outn(\typej)| = \begin{cases}
   2 &: \mbox{in a case of 1-1 or 1-2 of Lemma~\ref{lem:transition}} \\ 
   1 &:\mbox{otherwise}
  \end{cases}
\end{align}
  holds for any $\typej \in \typeset_{n-1}$, clearly $|\outn(q_0)| = 2$, and 
   $|\outn(\typei_n)|=|\outn(\flip{\typei}_n)|=0$ on $(Q_n,\delta)$ by definition.

  The following lemma is a strait-forward consequence of Lemma~\ref{lem:transition}. 
\begin{lemma}\label{lem:execution-path}
  ${\cal L}_n$ is bijective to the set of paths\footnote{
    A path may use an arc twice or more (i.e., a path may not be {\em simple}). 
    } starting from $q_0$ on $(Q_n \setminus \{\emptyset\},\delta)$. 
\end{lemma}

\subsection{Draw from ${\cal D}_n$ by a Markov chain on $Q_n$}\label{sec:MC}
 Let $X$ be a real-valued random variable drawn from $[0,1)$ uniformly at random. 
 Let ${\bf B}_{n+1} = B_1 \cdots B_{n+1} = \enc^{n+1}(X)$. 
 We here are concerned with the conditional probability 
\begin{align}
 \Pr[ B_{n+1}= b \mid \enc^n(X) = \bitseq_n]
\label{eq:cond_prob}
\end{align}
  for $\bitseq_n \in {\cal L}_n$ and $b=0,1$. 
 It is easy to see that 
\begin{align}
  \eqref{eq:cond_prob} 
   = \frac{|\se_{n+1}(\bitseq_n b)|}{|\se_n(\bitseq_n)|}
   = \frac{|\se_{n+1}(\bitseq_n b)|}{|\se_{n+1}(\bitseq_n 0)|+|\se_{n+1}(\bitseq_n 1)|}
  \label{eq:cond_prob2} 
\end{align}
   holds for $b=0,1$ 
 by \eqref{eq:subdiv-sec}. 
 Since a tent map $\tent$ is a piecewise linear function, 
  the following lemma seems intuitively and essentially trivial by Lemma~\ref{lem:transition}. 
 See Section~\ref{apx:cond-prob} for a proof. 
\begin{lemma}\label{lem:cond-prob}
Let $\bitseq_n \in {\cal L}_n$. Then, 
\begin{align*}
 \Pr[ B_{n+1}= b \mid \enc^n(X) = \bitseq_n]
    = \frac{|\typefn(\bitseq_n b)|}{|\typefn(\bitseq_n 0)|+|\typefn(\bitseq_n 1)|}
\end{align*}
 holds for $b \in \{0,1\}$, where let  $|\typefn(\bitseq_n b)| = 0$ if $\bitseq_n b \not\in {\cal L}_{n+1}$. 
\end{lemma}

We define a transition probability $p\colon Q_n \times Q_n \to \mathbb{R}_{\geq 0}$ as follows. 
Let
\begin{align}
  p(q_0, \typei_1) &= p(q_0, \flip{\typei}_1) = \frac{1}{2}. 
\label{eq:transition-prob1}
\end{align}
 For $\typej \in \typeset_{n-1}$, let
\begin{align}
  p(\typej, \delta(\typej, 0)) &= \frac{|\delta(\typej, 0)|}{|\delta(\typej, 0)| + |\delta(\typej, 1)|},& 
  p(\typej, \delta(\typej, 1)) &= \frac{|\delta(\typej, 1)|}{|\delta(\typej, 0)| + |\delta(\typej, 1)|}
\label{eq:transition-prob}
\end{align}
and $p(\typej, \typej') = 0$ for any $\typej'$ which is neither $\delta(\typej, 0)$ nor $\delta(\typej, 1)$.

 In fact, 
  $|\delta(\typej, 0)| + |\delta(\typej, 1)| =\mu|\typej|$ 
holds (see Lemma~\ref{lem:tasutomu}), and 
 the transition probability \eqref{eq:transition-prob} is rephrased by 
\begin{align}
  p(\typej, \typej') = \frac{|\typej'|}{\mu |\typej|}
\label{eq:transition-prob2}
\end{align}
 for any $\typej \in \typeset_n$ and $\typej'=\delta(\typej,b)$ ($b \in \{0,1\}$).

 Let $Z_0,Z_{1}, \ldots, Z_{n}$ be a Markov chain on $Q_n$ 
  according to \eqref{eq:transition-prob1} and \eqref{eq:transition-prob} where $Z_{0}=q_0$. 
 Let 
  $B'_i \in \{0,1\}$ for $i=1,2,\ldots,n$ be given by 
\begin{equation}
  B'_{i} =
  \begin{cases}
    0 &: Z_{i} = \delta(Z_{i-1}, 0) \\
    1 &: Z_{i} = \delta(Z_{i-1}, 1).
  \end{cases}
\label{def:Markov-bit}
\end{equation}

\begin{theorem}\label{thm:stat-distr}
The random bit sequence ${\bf B}'_n = B'_1\cdots B'_n$ given by \eqref{def:Markov-bit} follows ${\cal D}_n$. 
\end{theorem}
\begin{proof}
 Let ${\bf B}_n = B_1 \cdots B_n = \enc^n(X)$
 where $X$ is drawn from $[0,1)$ uniformly at random, 
  i.e.,  ${\bf B}_n$ follows ${\cal D}_n$. 
 Then, 
\begin{align*}
 \Pr[ {\bf B}'_n= \bitseq_n]
 &= \Pr[ B'_1= b_1 ] \prod_{i=1}^{n-1} \Pr[ B'_{i+1}= b_{i+1} \mid {\bf B}'_i = \bitseq_i] \\
 &= \frac{1}{2} \prod_{i=1}^{n-1} \frac{|\typefn(\bitseq_ib_{i+1})|}{|\typefn(\bitseq_i 0)|+|\typefn(\bitseq_i 1)|}
 && (\mbox{by \eqref{eq:transition-prob1} and \eqref{eq:transition-prob}})\\
 &= \Pr[ B_1= b_1 ] \prod_{i=1}^{n-1} \Pr[ B_{i+1}= b_{i+1} \mid {\bf B}_i = \bitseq_i] 
 && (\mbox{by Lemma~\ref{lem:cond-prob}})\\
 &=\Pr[ {\bf B}_n= \bitseq_n]
\end{align*}
 holds for any $\bitseq_n \in {\cal L}_n$. We obtain the claim. 
\end{proof}

\begin{algorithm}[t]
    \caption{Random generation of a tent code for a {\em rational}  $\mu = c/d \in (1,2)$}
    \label{alg1}
    \begin{algorithmic}[1]
    \REQUIRE a positive integer $n$ 
    \ENSURE a bit sequence $B_1\cdots B_n \sim {\cal D}_n$ 
    \STATE {\rm rational} $v[-1] \leftarrow 0$, {\rm rational} $u[-1] \leftarrow 0$ \COMMENT{$=\emptyset$}
    \STATE $v[0] \leftarrow 0$, $u[0] \leftarrow 1$, {\rm bit} $c[0] \leftarrow 0$, {\rm int} $\delta[0,0]\leftarrow 1$, $\delta[0,1]\leftarrow 1$ \COMMENT{$=q_0$}
    \STATE $v[1] \leftarrow 0$, $u[1] \leftarrow \tent(\frac{1}{2})$, $c[1] \leftarrow 1$, $\delta[1,0]\leftarrow 2$, $\delta[1,1]\leftarrow 1$   \COMMENT{$=\typei_1$}
    \STATE {\rm int} $k \leftarrow 1$, {\rm int} $l \leftarrow 0$, {\rm bit} $b \leftarrow 1$
    \FOR{$i=1$ to $n$}
    \IF[$Z_i =\typei_l$]{$b=c[l]$}
    \STATE bit $B \leftarrow 0$ w.p. $\frac{u[l'] - v[l']}{\mu(u[l]-v[l])}$ where $l'=\delta[l,0]$, {\bf otherwise} $B \leftarrow 1$
    \STATE $l \leftarrow \delta[l,B]$
    \ELSE[$Z_i = \flip{\typei}_l$]
    \STATE bit $B \leftarrow 1$ w.p. $\frac{u[l'] - v[l']}{\mu(u[l]-v[l])}$ where $l'=\delta[l,0]$, {\bf otherwise} $B \leftarrow 0$
    \STATE $l \leftarrow \delta[l,\flip{B}]$
    \ENDIF
    \RETURN $B$ \COMMENT{as $B_i$}
    \STATE $b \leftarrow B$
    \IF[``deferred update'' (ll.\ 15--27)]{$l = k$}
    \IF{$v[k] < \frac{1}{2} < u[k]$}
    \STATE $\delta[k,0] \leftarrow k+1$, $v[k+1] \leftarrow \tent(v[k])$, $u[k+1] \leftarrow \tent(\tfrac{1}{2})$, $c[k+1] \leftarrow 0$
\COMMENT{\footnotemark}
    \STATE $\delta[k,1] \leftarrow k'$ such that $v[k'] = \tent(u[k])$ and $u[k'] =\tent(\tfrac{1}{2})$
    \ELSIF{$u[k] \leq \frac{1}{2}$}
    \STATE $\delta[k,c[k]] \leftarrow k+1$, $v[k+1] \leftarrow \tent(v[k])$, $u[k+1] \leftarrow \tent(u[k])$, $c[k+1] \leftarrow c[k]$
    \STATE $\delta[k,\flip{c[k]}] \leftarrow -1$
    \ELSE[i.e., $v \lbrack k \rbrack \geq \frac{1}{2}$]
    \STATE $\delta[k,\flip{c[k]}] \leftarrow k+1$, $v[k+1] \leftarrow \tent(u[k])$, $u[k+1] \leftarrow \tent(v[k])$, $c[k+1] \leftarrow \flip{c[k]}$
    \STATE $\delta[k,c[k]] \leftarrow -1$
    \ENDIF
    \STATE $k \leftarrow k+1$
    \ENDIF
    \ENDFOR
    \end{algorithmic}
\end{algorithm}
\footnotetext{
  In fact, 
    we need ``$\delta[k,0] \leftarrow k'$ such that $v[k'] = \tent(u[k])$ and $u[k'] =\tent(\tfrac{1}{2})$'' after line 17, 
    similar to line 16 (see Theorem~\ref{theo:number_of_type}), 
     while it {\em seems} useless for {\em rational} $\mu$ (c.f., \cite{tomitamasterthesis,tomita2018}). 
  To avoid the bothering description, we here omit the operation. Lines 20 and 23 as well.  
  }

\subsection{Summary---The algorithm}\label{sec:algo-summary}
\subsubsection{Algorithm}
 Now, we give Algorithm~\ref{alg1} for Theorem~\ref{thm:main}, 
   by summarizing the arguments of Section~\ref{sec:algo}. 
 Basically,   
 the algorithm traces the Markov chain $Z_0,\ldots,Z_n$ and outputs $B_1,\ldots,B_n$  
    by $l$ and $b$ respectively,  
    according to the transition probability given in Section~\ref{sec:MC}. 

 In the algorithm, 
  $v[k]$ and $u[k]$ respectively denote  $\inf \typei_k$ and $\sup \typei_k$ for $k=1,2,\ldots$. 
 For descriptive purposes, 
  $v[0]$ and $u[0]$ corresponds to $q_0$, and  $v[-1]$ and $u[-1]$ corresponds to the reject state $\emptyset$ in $Q_n$. 
 The single bit $c[k]$ denotes $c_k$ for $\bitseqc_n =c_1 \cdots c_n = \enc^n(\frac{1}{2})$ (recall Theorem~\ref{theo:number_of_type}). 
 The pair $l$ and $b$ represent $Z_i = \typei_l$ if $b=c[l]$ (see line 6), 
  otherwise, i.e., $\flip{b} =c[l]$, $Z_i = \flip{\typei}_l$ (see line 9),
  at the $i$-th iteration (for $i=1,\ldots n$). 
To avoid the bothering notation, 
 we define the {\em level} of  $\typej \in \typeset_n$ 
  by 
\begin{align}
  L(\typej) = k
\label{def:level}
\end{align}
   if $\typej = \typei_k$ or $\flip{\typei_k}$,  
 where recall $\typei_k = \typefn(\bitseqc_k)$ and $\flip{\typei_k}= \typefn(\flip{\bitseqc_k})$
 for $\bitseqc_k = \enc^k(1/2)$ (see Theorem~\ref{theo:number_of_type}). 
 Then, $\delta(l,b) = l'$ represents the transition $\delta(\typej,b) = \typej'$ 
   given by \eqref{def:delta} in Section~\ref{sec:recog}
  (see also  \eqref{eq:transition1} and \eqref{eq:transition2}) 
 where $L(\typej) = l$ and $L(\typej') = l'$. 
  
 Lines 6--14 correspond to a transition from $Z_i$ to $Z_{i+1}$. 
 Algorithm~\ref{alg1} outputs every bit $B_i$ every iteration at line 7, 
  to avoid storing all $B_1 \cdots B_n$ that consumes $n$-bits of space.
 To attain the $\Order(\poly \log n)$ space for Theorem~\ref{thm:main}, 
   we use the deferred update strategy; 
  we calculate $v[k]$ and $u[k]$ representing $\typei_k$ on demand, 
   in lines 15--27 according to \eqref{eq:transition1} and \eqref{eq:transition2}. 
 By a standard argument of the space complexity of basic arithmetic operations, 
   see e.g., \cite{KV}, 
 rationals $v[k]$ and  $u[k]$ requires $\Order(k \log d)$ bits for each $k=1,2,\ldots$, 
   where the rational $\mu \in (1,2)$ is given by an irreducible fraction $c/d$. 

 Then, we look at the space complexity of the algorithm. 
 Rationals $v[k']$ and $u[k']$ ($k'=-1,0,1,2,\ldots,k$) consume at most $\Order(k^2 \log d)$ bits in total, 
    where $k$ denotes its value in the end of iterations of Algorithm~\ref{alg1}.  
 Integers $\delta[k',0]$ and $\delta[k',1]$ ($k'=0,1,2,\ldots,k$) consume at most $\Order(k \log k)$ bits in total. 
 Bits $c[k']$ ($k'=0,1,2,\ldots,k$) consume at most $\Order(k)$ bits in total.  
 Integers $k$, $l$ use $\Order (\log k)$ bits, and $b$ uses a single bit. 
 The value of $k$ becomes $n$ in the worst case, 
  while we will prove in Section~\ref{sec:average-complexity} that $k$ is $\Order(\log n \log d)$ in expectation, as well as with high probability.

\section{Average Space Complexity}\label{sec:average-complexity}
 Then, this section analyzes the average space complexity of Algorithm~\ref{alg1}, and proves Theorem~\ref{thm:main}. 
 As we stated in Section~\ref{sec:algo-summary}, 
   our goal of the section is essentially to prove that 
  the maximum value of $k$ of Algorithm~\ref{alg1} in the $n$ iterations  
  is $\Order(\log n)$ with high probability, and in expectation. 
 More precisely, 
  let $Z_0,\ldots,Z_n$ be a Markov chain 
   on {\em segment-types}
   with $Z_0=q_0$ given in Section~\ref{sec:MC}.
 Let  
\begin{align}
  \requiredstgsize=\min\{ k \in \mathbb{Z}_{>0} \mid L(Z_i) = k \}
\label{def:K}
\end{align}  
  be a random variable,  
   where $L(Z_i) = k$ denotes $Z_i = \typei_k$ or $\flip{\typei}_k$ 
   (recall \eqref{def:level} as well as Theorem~\ref{theo:number_of_type}). 
 Note that $K \leq n$, that is a trivial upper bound. 
 We want $\E[K] = \Order(\log n)$. 
 The following fact is the key for the purpose. 
\begin{lemma}\label{lemm:transition_function}
 Suppose for $\mu \in (1, 2)$ that $\tent_{\mu}^{i}(\frac{1}{2}) \neq \frac{1}{2}$ holds for any $i = 1, \dots, n-1$. 
 Then, 
 \begin{align}
   \delta(\typei_{n}, b) &\in \left\{\typei_{n+1}\right\} \cup \left\{\flip{\typei}_{k+1} \mid 1 \le k \le \tfrac{n}{2} \right\} \cup \{\emptyset\}
   \label{eq:transition_function_1} \\
   \delta(\flip{\typei}_{n}, b) &\in \left\{\flip{\typei}_{n+1}\right\} \cup  \left\{\typei_{k+1} \mid 1 \le k \le \tfrac{n}{2} \right\} \cup \{\emptyset\}
   \label{eq:transition_function_2}
 \end{align}
   hold for $b=0,1$. 
\end{lemma}
 We will prove Lemma~\ref{lemm:transition_function} in Section~\ref{sec:transition_function}. 
 Roughly speaking, Lemma~\ref{lemm:transition_function} implies that 
  the level $L(Z_i)$ increases by one, or decreases into (almost) a half 
  by a step of the Markov chain. 
 Thus, the issue is the ratio of the transition. 
 A precise argument follows. 

\subsection{Proof  of Theorem~\ref{thm:main}}\label{sec:rational}
 In fact, 
   we will prove $\E[K^2] = \Order(\log^2 n \log^2 d)$ in Lemma~\ref{lemm:average_space_complexity}, 
   since what we really want for Theorem~\ref{thm:main} is the expected space complexity, instead of $\E[K]$ itself. 
 Formally, the following lemma claims that 
  the description of the Markov chain on $Q_k$ requires $\Order(k^2 \log d)$ space (see also Section~\ref{sec:algo-summary}). 

\begin{lemma}
  \label{lemm:space_of_stg}
  Let $\mu \in (1, 2)$ be rational given by an irreducible fraction $\mu = c/d$.
  For any $k \in \mathbb{Z}_{>0}$,
  the Markov chain on $Q_k$ is represented by $\Order(k^{2} \lg{d})$ bits. 
\end{lemma}

\begin{proof}
 As we stated in Section~\ref{sec:algo-summary}, 
  we can describe a segment-type 
   by two numbers representing the interval ($v[k]$ and $u[k]$ in Algorithm~\ref{alg1}), and 
   a single bit to identify $\typei_k$ or $\flip{\typei}_k$ ($c[k]$, there), by Lemma~\ref{lem:hanten-type}. 
 We will prove Lemma \ref{lemm:elements_of_type_2} 
   which claims that $\typei_k$ ($\flip{\typei}_k$ as well) is 
   either $[\tent^{i}(\frac{1}{2}), \tent^{j}(\frac{1}{2}))$ or $(\tent^{i}(\frac{1}{2}), \tent^{j}(\frac{1}{2})]$ with $i \le k$ and $j \le k$.
  For any $i$, there exists $c_{i} \in \mathbb{Z}$ satisfies $\tent^{i}(\frac{1}{2}) = \frac{c_{i}}{2 d^{i}}$ and $0 \le c_{i} \le 2 d^{i}$ since $0 \le \tent^{i}(\frac{1}{2}) \le 1$.
  Thus, every segment-type $\typej \in \typeset_k$ requires at most $4(k \log_{2} d + 1) + 1$ bits. 
  Since $Q_k$ contains $2k$ segment-types with at most $4k$ arcs,  
   the Markov chain on $Q_k$ is represented by $\Order(k^{2} \log{d})$ bits.
\end{proof}

Notice that $\E[\requiredstgsize^{2} \log d] = \log d\,\E[\requiredstgsize^{2}]$ holds. 
 We prove $\E[\requiredstgsize^{2}] = \Order(\log^{2}{n} \log^{2}{d})$ for the Markov chain, 
 which generates a random bit sequence $B_1\cdots B_n \sim {\cal D}_n$. 
\begin{lemma}
  \label{lemm:average_space_complexity}
  Let $\mu \in (1, 2)$ be rational given by an irreducible fraction $\mu = c/d$.
  Suppose\footnote{
    This assumption may be redundant by the assumption of $\mu$ rational. 
    It needs for Lemma~\ref{lemm:transition_function}. 
  } for $\mu \in (1, 2)$ that $\tent_{\mu}^{i}(\frac{1}{2}) \neq \frac{1}{2}$ holds for any $i = 1, \dots, n-1$. 
  Then,
    $\E[\requiredstgsize^{2}] = \Order(\log_{\mu}^{2}{n} \log_{\mu}^{2}{d}) = \Order(\lg^2 n \lg ^2 d / \lg^4 \mu)$. 
\end{lemma}
  It is not essential that 
   Lemma~\ref{lemm:average_space_complexity} assumes $\mu$ rational, 
  which just follows that of Theorem~\ref{thm:main} for an argument about Turing comparability. 
 We will establish a similar (but a bit weaker) Proposition~\ref{prop:real-prob} 
  for any real $\mu \in (0,1)$ in Section~\ref{sec:real}. 
 Before the proof of Lemma~\ref{lemm:average_space_complexity}, 
  we prove Theorem~\ref{thm:main} by Lemmas~\ref{lemm:space_of_stg} and \ref{lemm:average_space_complexity}. 
\begin{proof}[Proof  of Theorem~\ref{thm:main}]
  Algorithm~\ref{alg1} provides 
    a random bit sequence $B_1\cdots B_n \sim {\cal D}_n$, by Theorem~\ref{thm:stat-distr}. 
  Lemma \ref{lemm:space_of_stg} implies 
   that its space complexity is $\Order(\requiredstgsize^{2} \lg{d})$. 
 Lemma \ref{lemm:average_space_complexity} implies 
  \begin{align*}
    \E[\Order(\requiredstgsize^{2} \lg{d})]
    &= \Order(\E[\requiredstgsize^{2}] \lg{d}) \\
    &= \Order(\lg^{2}{n} \lg^{3}{d} / \lg^4 \mu)  
  \end{align*}
  and we obtain the claim. 
\end{proof}
 We remark that the space complexity is clearly $\Order(1)$ to $n$ 
    if there exists a constant $n_*$ for $\mu \in (1,2)$ such that $|\typeset_{n,\mu}| \leq 2n_*$ holds for any $n \geq n_*$, 
   meaning that $|Q_n| = |\typeset_n|+2$ is constant to $n$ 
   (recall Theorem~\ref{theo:number_of_type}). 
 In the case\footnote{
   It seems not the case for any {\em rational} $\mu$, 
   though we do not prove it here (cf.\ \cite{tomitamasterthesis,tomita2018}).}, 
 we can prove the following fact (cf.\ the hypothesis of Lemma~\ref{lemm:transition_function} 
  caused by  Lemma~\ref{lemm:average_space_complexity}). 
 See Section~\ref{sec:transition_function} for a proof of Proposition~\ref{prop:const}. 
\begin{proposition}\label{prop:const}
 Let $\mu \in (1,2)$. 
 If there exists $n$ ($n=1,2,\ldots$) such that $\tent_{\mu}^n(\frac{1}{2}) = \frac{1}{2}$ holds 
   then $|\typeset_{n'}| \leq 2n$ holds for any $n' \geq n$. 
\end{proposition}
 The rest of Section~\ref{sec:average-complexity} proves 
  Lemmas~\ref{lemm:average_space_complexity} and \ref{lemm:transition_function} 
  in Sections~\ref{sec:proof-lemm5.2} and \ref{sec:transition_function}, respectively.

\subsection{Proof of Lemma ~\ref{lemm:average_space_complexity}}\label{sec:proof-lemm5.2}
 The proof strategy of  Lemma~\ref{lemm:average_space_complexity} is as follows. 
Our key fact 
 Lemma~\ref{lemm:transition_function}
   implies that a chain must follow the path $\typei_l, \typei_{l+1}, \ldots, \typei_{2l}$ 
     (or $\flip{\typei}_l, \flip{\typei}_{l+1}, \ldots, \flip{\typei}_{2l}$) to reach level $2l$ 
    and the probability is $\frac{|\typei_{2l}|}{\mu^{l}|\typei_{l}|}$ (Lemma~\ref{lem:go_back}). 
 We then prove that 
    there exists $l=\Order(\log n \log d)$ such that $\frac{|\typei_{2l}|}{\mu^{l}|\typei_{l}|} \leq n^{-3}$
      (Lemma~\ref{lemm:existence_of_short_q_s}), 
  which provides $\Pr[K \geq 2l] \leq n^{-2}$ (Lemma~\ref{lem:l*}).  
 Lemma~\ref{lemm:average_space_complexity} is easy from  Lemma~\ref{lem:l*}.

 We observe the following fact from Lemma~\ref{lemm:transition_function}. 
\begin{observation}\label{obs:go_straight}
  If $Z_t$ visits $\typei_{2j}$ (resp.\ $\flip{\typei}_{2j}$) {\em for the first time} then 
    $Z_{t-i} = \typei_{2j-i}$ (resp.\ $Z_{t-i} = \flip{\typei}_{2j-i}$) for $i=1,2,\ldots,j$. 
\end{observation}
\begin{proof}
 By Lemma~\ref{lemm:transition_function}, 
  all in-edges to $\typei_k$ (resp.\ $\flip{\typei}_k$) for any $k = j+1,\dots,2j$ 
  come from $\typei_{k-1}$ (resp.\ $\flip{\typei}_{k-1}$),  or a node of level $2j$ or greater. 
 Since $Z_t$ has not visited any level greater than $2j$ by the hypothesis and the above argument again, 
 we obtain the claim. 
\end{proof}

 By Observation~\ref{obs:go_straight}, 
   if a Markov chain $Z_1,Z_2,\ldots$ visits level $2l$ {\em for the first time} at time $t$ 
     then $\lev(Z_{t-l})$ must be $l$. 
 The next lemma gives an upper bound of the probability from level $l$ to $2l$. 
\begin{lemma}\label{lem:go_back}
$\Pr[\lev(Z_{t}) = 2l \mid \lev(Z_{t-l}) = l] 
 =  \frac{|\typei_{2l}|}{\mu^{l} |\typei_{l}|}$. 
\end{lemma}
\begin{proof}
 By Observation~\ref{obs:go_straight}, 
  the path from $\typei_l$ to  $\typei_{2l}$ is unique and 
  \begin{align}
    \label{eq:asc_1_7}
 \Pr[Z_{t}  = \typei_{2l} \mid Z_{t-l} = \typei_l] 
    &= \prod_{i=l}^{2l-1} p(\typei_{i}, \typei_{i+1}) \nonumber \\
    &= \prod_{i = l}^{2l-1} \frac{|\typei_{i+1}|}{\mu |\typei_{i}|} && (\mbox{by \eqref{eq:transition-prob2}}) \nonumber \\
    &= \frac{|\typei_{2l}|}{\mu^{l} |\typei_{l}|}
  \end{align}
holds. 
 We remark that $|\typei_i| = |\flip{\typei}_i|$ holds for any $i$, meaning that 
  $p(\typei_{i}, \typei_{i+1}) = p(\flip{\typei}_{i}, \flip{\typei}_{i+1})$, 
   and hence 
  $\Pr[Z_{t}  =  \flip{\typei}_{2l} \mid Z_{t-l} = \flip{\typei}_l] = \frac{|\typei_{2l}|}{\mu^{l} |\typei_{l}|}$. 
\end{proof}

 The following lemma is the first mission of the proof of Lemma~\ref{lemm:average_space_complexity}. 
\begin{lemma}
  \label{lemm:existence_of_short_q_s}
  Let $\mu \in (1, 2)$ be rational given by an irreducible fraction $\mu = c/d$.
 Suppose for $\mu \in (1, 2)$ that $\tent_{\mu}^{i}(\frac{1}{2}) \neq \frac{1}{2}$ holds for any $i = 1, \dots, n-1$. 
 Then, there exists $l$ such that  $l \leq 8 \ceil{\log_{\mu} d} \ceil{ \log_{\mu} n} $ and 
  \begin{equation}
    \label{eq:shortqs_0}
    \frac{|\typei_{2l}|}{\mu^l |\typei_l|} \le n^{-3}
  \end{equation}
  holds. 
\end{lemma}
To prove Lemma \ref{lemm:existence_of_short_q_s}, 
 we remark the following fact. 
\begin{lemma}
  \label{lemm:min_sectype_length}
  Let $\mu \in (1, 2)$ be rational given by an irreducible fraction $\mu = c/d$.
  Then, $|\typei_k| \ge \frac{1}{2 d^k}$ for any $k \ge 2$. 
\end{lemma}

\begin{proof}
 We will prove Lemma \ref{lemm:elements_of_type_2} in Section~\ref{sec:transition_function}, 
  which claims that $\typei_k$ is either 
   $\left[\tent^{i}(\frac{1}{2}), \tent^{j}(\frac{1}{2})\right)$ or $\left(\tent^{i}(\frac{1}{2}), \tent^{j}(\frac{1}{2})\right]$ where $i \le k$ and $j \le k$.
  We can denote $\tent^{i}(\frac{1}{2})$ as $\frac{c_{i}}{2d^{i}}$ ($c_{i} \in \mathbb{Z}_{>0}$) for any $i$.
  Therefore,
  \begin{equation}
    |\typei_k|
    = \left|\tent^{i}(\tfrac{1}{2}) - \tent^{j}(\tfrac{1}{2})\right|
    = \left|\frac{c_{i}}{2 d^{i}} - \frac{c_{j}}{2 d^{j}}\right|
    = \left|\frac{c_{i} d^{k-i} - c_{j} d^{k-j}}{2 d^k}\right|
  \end{equation}
  holds.
  Clearly, $c_{i} d^{k-i} - c_{j} d^{k-j}$ is an integer, and it is not $0$ since $|\typei_k| \neq 0$.
  Thus, we obtain $|\typei_k| \ge \frac{1}{2 d^k}$.
\end{proof}

 Then, we prove Lemma~\ref{lemm:existence_of_short_q_s}. 
\begin{proof}[Proof of Lemma~\ref{lemm:existence_of_short_q_s}]
  For convenience, let 
   $l_i =  2^i \ceil{\log_{\mu}n} $ for $i = 1, 2, \ldots$. 
  Assume for a contradiction that 
     \eqref{eq:shortqs_0} never hold for any $l_1,l_2,\ldots,l_k$, where 
  $k = \max\{4, \ceil{\log_{2} \log_{\mu} d} + 2\}$ for convenience.  
 In other words, 
  \begin{equation}
    |\typei_{l_{i+1}}| > n^{-3} \mu^{l_{i}} |\typei_{l_{i}}|
  \end{equation}
  holds every $i=1,2,\ldots,k$. 
 Thus, we inductively obtain that 
  \begin{align}
    |\typei_{l_{k+1}}|
    &> n^{-3} \mu^{l_{k}} |\typei_{l_{k}}| \nonumber \\
    &> n^{-6} \mu^{l_{k}} \mu^{l_{k-1}} |\typei_{l_{k-1}}| \nonumber \\
    &> \dots \nonumber \\
    &> n^{-3k} \mu^{l_{k}} \mu^{l_{k-1}} \dots \mu^{l_1} |\typei_{l_1}|
    \label{eq:shortqs_2}
  \end{align}
   holds. 
 By the definition of $l_i$, 
  \begin{equation}
    \label{eq:shortqs_5}
    \mu^{l_{i}} = \mu^{2^{i} \ceil{\log_{\mu}n}} \ge \mu^{2^{i} \log_{\mu}n} = n^{2^{i}}
  \end{equation}
  holds. 
 Lemma \ref{lemm:min_sectype_length} implies that 
  \begin{align}
    \label{eq:shortqs_7}
    |\typei_{l_{1}}|
    \ge \frac{1}{2 d^{l_{1}}}  
    = \frac{1}{2} d^{-2 \ceil{\log_{\mu} n}} 
    \geq \frac{1}{2} d^{-4 \log_{\mu} n} 
    = \frac{1}{2} n^{-4 \log_{\mu} d}
  \end{align}
 holds.  
 Then,  \eqref{eq:shortqs_2},  \eqref{eq:shortqs_5} and \eqref{eq:shortqs_7} imply 
  \begin{align}
    |\typei_{l_{k+1}}|
    &>  n^{-3k}  \cdotp  n^{2^{k} + 2^{k-1} + \dots + 2^{1}} \cdotp \frac{1}{2} n^{-4 \log_{\mu} d} 
    \label{eq:shortqs_6}
  \end{align}
 holds. 
  By taking the $\log_{n}$ of the both sides of \eqref{eq:shortqs_6}, 
    we see that 
  \begin{align}
    \log_{n}{|\typei_{l_{k+1}}|}
    &> -3k + 2^{k+1} - 2 - \log_{n}2 - 4 \log_{\mu}d  \nonumber \\
    &= (2^{k} - 4 \log_{\mu}d) + (2^{k} -3k - 2 - \log_{n}2)
    \label{eq:shortqs_8}
  \end{align}
   holds. 
 Since  $k \geq  \ceil{\log_{2} \log_{\mu} d} + 2$ by definition,  
  it is not difficult to see  that 
  \begin{align}
    2^{k} - 4 \log_{\mu}d 
    &\geq 2^{2 + \log_{2}(\log_{\mu}d)} - 4 \log_{\mu}d \nonumber \\
    &= 4 \log_{\mu}d - 4 \log_{\mu}d \nonumber \\
    &= 0
    \label{eq:shortqs_1}
  \end{align}
  holds.
 Since  $k \geq  4$ by definition,  
 it is also not difficult to observe that 
  \begin{align}
    2^{k} -3k - 2 - \log_{n} 2
   \ \ge\ 2^{4} - 3 \cdot 4 - 2 - \log_{n} 2 
   \ =\ 2- \log_{n} 2
   \ >\ 0
    \label{eq:shortqs_9}
  \end{align}
  holds.  
  Equations \eqref{eq:shortqs_8}, \eqref{eq:shortqs_1} and \eqref{eq:shortqs_9} imply that $\log_{n}{|\typei_{l_{k+1}}|} > 0$, 
     meaning that $|\typei_{l_{k+1}}| > 1$. 
  At the same time, 
    notice that any segment-type $\type$ satisfies $\type \subseteq [0,1]$,  
     meaning that 
    $|\typei_{l_{k+1}}| \le 1$. 
  Contradiction. 
 Thus, we obtain \eqref{eq:shortqs_0} for at least one of $l_1,l_2,\ldots,l_k$.  

 Finally, we check the size of $l_k$:  
\begin{align*}
  l_k &= 2^k \ceil{\log_{\mu} n} \leq 2^{\max\{4,\ceil{\log_2 \log_{\mu} d} + 2\}} \ceil{\log_{\mu} n} 
  \leq 2^{\max\{4, \log_2 \log_{\mu} d + 3\}} \ceil{\log_{\mu} n} \\& = \max\{ 16, 8 \log_{\mu} d \} \ceil{\log_{\mu} n}  = 8 \max\{ 2, \log_{\mu} d \} \ceil{\log_{\mu} n} \leq 8 \ceil{\log_{\mu} d} \ceil{\log_{\mu} n} 
\end{align*}
 where the last equality follows $\log_{\mu} d > 1$ since $\mu < 2$ and $d \geq 2$. 
 We obtain a desired $l$. 
\end{proof}

By Lemmas~\ref{lem:go_back} and \ref{lemm:existence_of_short_q_s}, we obtain the following fact. 
\begin{lemma}\label{lem:l*}
Let $l_* = 8 \ceil{\log_{\mu} d} \ceil{\log_{\mu} n} $ for convenience. Then 
\begin{align*}
 \Pr[\requiredstgsize \ge 2 l_*]  \leq n^{-2} 
\end{align*} 
 holds. 
\end{lemma}
\begin{proof}
 For $\mu = c/d$ and $n$, Lemma~\ref{lemm:existence_of_short_q_s} implies that there exists $l$ such that 
  $l \leq l_* $ and 
\begin{align}
 \frac{|\typei_{2l}|}{\mu^{l} |\typei_{l}|} \leq n^{-3}
\label{eq:l*1}
\end{align}
holds.  
 Let $A_t$ ($t=1,\ldots,n$) denote the event that $Z_t$ reaches the level $2 l$ {\em for the first time}. 
 It is easy to see that 
\begin{align}
 \Pr[\requiredstgsize \ge 2l] 
  = \Pr\left[\bigvee_{t=0}^{n} A_t \right] 
\label{eq:l*2}
\end{align}
 holds\footnote{Precisely, $\bigvee_{t=0}^{n} A_t = \bigvee_{t=2l_*}^{n} A_t$ holds, but we do not use the fact here.  } 
   by the definition of $A_t$. 
 We also remark that the event $A_t$ implies not only $Z_t = 2l$ but also $\lev(Z_{t-l}) = l$ by Observation~\ref{obs:go_straight}.  
 It means that 
\begin{align}
 \Pr[A_t] 
  &\leq \Pr[ [\lev(Z_t) = 2l] \wedge [\lev(Z_{t-l}) = l] ] 
  \nonumber\\
  &= \Pr[ \lev(Z_t) = 2l  \mid \lev(Z_{t-l}) = l ] \Pr[ \lev(Z_{t-l}) = l ] \nonumber\\
  &\leq \Pr[ \lev(Z_t) = 2l  \mid \lev(Z_{t-l}) = l ] 
\label{eq:l*3}
\end{align}  
 holds. Then, 
\begin{align*}
 \Pr[\requiredstgsize \ge 2l] 
 &= \Pr\left[\bigvee_{t=0}^{n} A_t \right] 
   && (\mbox{by \eqref{eq:l*2}}) \nonumber \\
 &\le \sum_{t=0}^n \Pr\left[A_t\right] 
   && (\mbox{union bound})  \nonumber \\
  &\leq n \Pr[ \lev(Z_t) = 2l  \mid \lev(Z_{t-l}) = l ] 
   && (\mbox{by \eqref{eq:l*3}}) \nonumber \\
  &\leq n \frac{|\typei_{2l}|}{\mu^{l} |\typei_{l}|}
   && (\mbox{by Lemma~\ref{lem:go_back}}) \nonumber \\
  & \leq n^{-2}
   && (\mbox{by \eqref{eq:l*1}}) 
\end{align*}
  holds. 
We remark that 
$\Pr[\requiredstgsize \ge 2l_*] 
 \leq 
 \Pr[\requiredstgsize \ge 2l] $
is trivial since $l < l_*$. 
\end{proof}

We are ready to prove Lemma \ref{lemm:average_space_complexity}. 
\begin{proof}[Proof of Lemma \ref{lemm:average_space_complexity}]
Let $l_* = 8 \ceil{\log_{\mu} d} \ceil{\log_{\mu} n} $ for convenience. Then 
\begin{align*}
  \E[\requiredstgsize^{2}]
   &= \sum_{k=1}^n k^2\Pr[\requiredstgsize =k]  \nonumber \\
   &= \sum_{k=1}^{2l_*-1} k^2\Pr[\requiredstgsize =k]  + \sum_{k=2l_*}^n k^2\Pr[\requiredstgsize =k]  \nonumber \\
   &\leq (2l_* - 1)^2 \Pr[\requiredstgsize \leq 2l_*-1] + n^{2} \Pr[\requiredstgsize \ge 2l_*] \nonumber \\
   &\leq (2l_* - 1)^2 + n^{2} \Pr[\requiredstgsize \ge 2l_*] \nonumber \\
   &\leq (2l_* - 1)^2 + 1
    && (\mbox{by Lemma~\ref{lem:l*}}) \\
   &= (16 \ceil{\log_{\mu} d} \ceil{\log_{\mu} n} -1)^2 + 1
  \end{align*}
holds. Now the claim is easy. 
\end{proof}

\subsection{Proofs of Lemmas~\ref{lemm:transition_function} and \ref{prop:const}}\label{sec:transition_function}
 This section proves Lemmas~\ref{lemm:transition_function} and \ref{prop:const}. 
 \begin{lemma}[Lemma~\ref{lemm:transition_function}]
 Suppose for $\mu \in (1, 2)$ that $\tent^{i}(\frac{1}{2}) \neq \frac{1}{2}$ holds for any $i = 1, \dots, n-1$. 
 Then, 
 \begin{align}
   \delta(\typei_{n}, b) &\in \left\{\typei_{n+1}\right\} \cup \left\{\flip{\typei}_{k+1} \mid 1 \le k \le \tfrac{n}{2} \right\} \cup \{\emptyset\}
   \tag{\ref{eq:transition_function_1}} \\
   \delta(\flip{\typei}_{n}, b) &\in \left\{\flip{\typei}_{n+1}\right\} \cup  \left\{\typei_{k+1} \mid 1 \le k \le \tfrac{n}{2} \right\} \cup \{\emptyset\}
   \tag{\ref{eq:transition_function_2}}
 \end{align}
   hold for $b=0,1$. 
\end{lemma}

 To begin with, we give two remarks. 
 One is that 
  we know $\typei_{n+1} = \delta(\typei_n,0)$ if $|\outn(\typei_n)|=2$, 
  otherwise $\delta(\typei_n,c_{n+1}) = \typei_{n+1}$ and $\delta(\typei_n,\flip{c_{n+1}}) = \emptyset$, by Lemma~\ref{lem:transition}.  
 The other is that $L(\delta(\typei_n,b)) = L(\delta(\flip{\typei}_n,\flip{b}))$ holds 
  since $\typefn(\bitseq_n) = \typefn(\flip{\bitseq_n})$ holds for any $\bitseq_n \in {\cal L}_n$ by Lemma~\ref{lem:hanten-type}. 
 Thus we only need to prove the following lemma as a proof of Lemma~\ref{lemm:transition_function}
\begin{lemma}\label{lemm:transition_function2}
 Suppose for $\mu \in (1, 2)$ that $\tent^{i}(\frac{1}{2}) \neq \frac{1}{2}$ holds for any $i = 1, \dots, n-1$. 
 If $|\outn(\typei_n)|=2$ holds for $n$ ($n \geq 2$) then 
 \begin{align}
   \delta(\typei_{n}, 1) &\in  \left\{\flip{\typei}_{k+1} \mid 1 \le k \le \frac{n}{2} \right\}
 \end{align}
   holds. 
\end{lemma}
We will prove Lemma~\ref{lemm:transition_function2}.
For the purpose, we define
\begin{equation}
  \label{eq:partner_1}
  \partner_{\mu}(n) \defeq
  \min \left\{ i \in \{1, 2, \dots, n\} \mid |\outn(\typei_{n-i}) | = 2 \right\}
\end{equation}
for $n = 2,3,\ldots$, 
 where we also use $\partner(n)$ as $\partner_{\mu}(n)$ without a confusion. 
Notice that $\partner(2) = 1$ for any $\mu \in (1,2)$ since $|\outn(\typei_1) | = |\{[0,\frac{\mu}{2}),(0,\frac{\mu}{2}]\}|= 2$. 
We also remark a recursion 
\begin{equation}
  \label{eq:partner_2}
  \partner(n+1) =
  \begin{cases}
    1 &: |\outn(\typei_n) | = 2,\\
    \partner(n) + 1 &: |\outn(\typei_n) | = 1 
  \end{cases}
\end{equation}
holds by the definition. 
Then, we give a refinement of Lemma~\ref{lem:transition} (see also Theorem~\ref{theo:number_of_type}). 
\begin{lemma}\label{lemm:elements_of_type_2}
  For any $n \ge 2$,
  \begin{equation}
    \label{eq:elements_of_type_b}
    \typei_n=
    \begin{cases}
      [\tent^{n}(\frac{1}{2}), \tent^{\partner(n)}(\frac{1}{2}))
      &\mbox{if $c_{n} = 0$,}\\
      (\tent^{\partner(n)}(\frac{1}{2}), \tent^{n}(\frac{1}{2})]
      &\mbox{otherwise}
    \end{cases}
  \end{equation}
  holds where recall $\typei_n=\typefn(\bitseqc_{n})$. 
\end{lemma}
 We remark that 
  \begin{equation}
    \label{eq:elements_of_type_b_2}
    \flip{\typei}_n=
    \begin{cases}
      (\tent^{n}(\frac{1}{2}), \tent^{\partner(n)}(\frac{1}{2})]
      &\mbox{if $c_{n} = 0$,}\\
      [\tent^{\partner(n)}(\frac{1}{2}), \tent^{n}(\frac{1}{2}))
      &\mbox{otherwise}
    \end{cases}
  \end{equation}
  holds by Lemmas~\ref{lemm:elements_of_type_2} and \ref{lem:hanten}, 
  where recall $\flip{\typei}_n=\typefn(\flip{\bitseqc_{n}})$. 
\begin{proof}[Proof of Lemma~\ref{lemm:elements_of_type_2}]
 The proof is an induction on $n$. 
 For $n=2$,  
  notice that $\bitseqc_{2} = \enc^2(\frac{1}{2}) = 10$ for any $\mu \in (1,2)$. 
 Then, 
  \begin{align*}
    \typei_2=\typefn(10) &= [\tent^{2}(\tfrac{1}{2}), \tent^{}(\tfrac{1}{2})), 
  \end{align*}
  which  implies 
  \eqref{eq:elements_of_type_b} for $n=2$, 
  where we remark that  $\partner(2)=1$ (see \eqref{eq:partner_1}). 

 Inductively assuming \eqref{eq:elements_of_type_b} holds for $n$, we prove it for $n+1$. 
 We consider the cases $c_{n} = 0$ or~$1$. 
 Firstly, we are concerned with the case of $c_n=0$. 
 In the case, 
  $\typefn(\bitseqc_{n}) = [\tent^{n}(\frac{1}{2}), \tent^{\partner(n)}(\frac{1}{2}))$ 
  by the inductive assumption.
 We consider the following three cases. 
   \vspace{-1ex}
\begin{itemize}\setlength{\itemindent}{2em}\setlength{\parskip}{0.5ex}\setlength{\itemsep}{0cm} 
 \item[Case 1-1.]
 Suppose that $\tent^{n}(\frac{1}{2}) < \frac{1}{2} < \tent^{\partner(n)}(\frac{1}{2})$. 
 Then, $|\outn(\typei_n)| =2$. 
  Recall that 
    $\bitseqc_{n+1} = \min\{ \bitseq_{n+1} \in {\cal L}_n \mid b_1 = 1  \mbox{ where } \bitseq_n  = b_1 \cdots b_{n+1} \}$ by Lemma~\ref{lem:c=d}, 
   thus $c_{n+1}$ must be $0$. 
  Recall Case 1-1-1 in Lemma~\ref{lem:transition}, then 
    \begin{equation*}
      \typefn(\bitseqc_{n+1})
      = \typefn(\bitseqc_{n} 0)
      = [\tent^{n+1}(\tfrac{1}{2}), \tent^{1}(\tfrac{1}{2}))
      = [\tent^{n+1}(\tfrac{1}{2}), \tent^{\partner(n+1)}(\tfrac{1}{2}))
    \end{equation*}
    holds where we used $\partner(n+1)=1$ by \eqref{eq:partner_2}.
  We obtain the claim in the case. 
 \item[Case 1-2.]
    Suppose that $\tent^{\partner(n)}(\frac{1}{2}) \le \frac{1}{2}$. 
    Then, $|\outn(\typei_n)| =1$. 
    Recall Case 1-2 of Lemma~\ref{lem:transition}, then 
    \begin{equation*}
      \typefn(\bitseqc_{n+1})
      = [\tent^{n+1}(\tfrac{1}{2}), \tent^{\partner(n)+1}(\tfrac{1}{2}))
      = [\tent^{n+1}(\tfrac{1}{2}), \tent^{\partner(n+1)}(\tfrac{1}{2}))
    \end{equation*}
  holds where we used $\partner(n+1)=\partner(n)+1$ by \eqref{eq:partner_2}. 
  We obtain the claim in the case. 
    \item[Case 1-3.]
    Suppose that $\tent^{n}(\frac{1}{2}) \ge \frac{1}{2}$.
    Then, $|\outn(\typei_n)| =1$. 
    Recall Case 1-3 of Lemma~\ref{lem:transition}, then 
    \begin{equation*}
      \typefn(\bitseqc_{n+1})
      = (\tent^{\partner(n)+1}(\tfrac{1}{2}), \tent^{n+1}(\tfrac{1}{2})]
      = (\tent^{\partner(n+1)}(\tfrac{1}{2}), \tent^{n+1}(\tfrac{1}{2})]
    \end{equation*}
   holds where we used $\partner(n+1)=\partner(n)+1$ by \eqref{eq:partner_2}.
  We obtain the claim in the case. 
  \end{itemize}
Next, suppose $c_{n} = 1$. 
  Then $\typefn(\bitseqc_{n}) = (\tent^{\partner(n)}(\frac{1}{2}), \tent^{n}(\frac{1}{2})]$.
  \vspace{-1ex}
\begin{itemize}\setlength{\itemindent}{2em}\setlength{\parskip}{0.5ex}\setlength{\itemsep}{0cm} 
 \item[Case 2-1.]
    Suppose that $\tent^{\partner(n)}(\frac{1}{2}) < \frac{1}{2} < \tent^{n}(\frac{1}{2})$. 
    Then, $|\outn(\typei_n)| =2$. 
  Recall that 
    $\bitseqc_{n+1} = \min\{ \bitseq_{n+1} \in {\cal L}_n \mid b_1 = 1  \mbox{ where } \bitseq_n  = b_1 \cdots b_{n+1} \}$ by Lemma~\ref{lem:c=d}, 
   thus $c_{n+1}$ must be $0$. 
    Recall Case 2-1-2 of Lemma~\ref{lem:transition}, then 
    \begin{equation*}
      \typefn(\bitseqc_{n+1})
      = \typefn(\bitseqc_{n} 0)
      = [\tent^{n+1}(\tfrac{1}{2}), \tent^{1}(\tfrac{1}{2}))
      = [\tent^{n+1}(\tfrac{1}{2}), \tent^{\partner(n+1)}(\tfrac{1}{2}))
    \end{equation*}
    holds where we used $\partner(n+1)=1$ by \eqref{eq:partner_2}.
  We obtain the claim in the case. 
 \item[Case 2-2.]
    Suppose that $\tent^{n}(\frac{1}{2}) \le \frac{1}{2}$.
    Then, $|\outn(\typei_n)| =1$. 
    Recall Case 2-2 of Lemma~\ref{lem:transition}, then 
    \begin{equation*}
      \typefn(\bitseq_{n+1})
      = \typefn(\bitseq_{n} b_{n})
      = (\tent^{\partner(n)+1}(\tfrac{1}{2}), \tent^{n+1}(\tfrac{1}{2})]
      = (\tent^{\partner(n+1)}(\tfrac{1}{2}), \tent^{n+1}(\tfrac{1}{2})]
    \end{equation*}
   holds where we used $\partner(n+1)=\partner(n)+1$ by \eqref{eq:partner_2}.
  We obtain the claim in the case. 
 \item[Case 2-3.]
    Suppose that $\tent^{\partner(n)}(\frac{1}{2}) \ge \frac{1}{2}$.
    Then, $|\outn(\typei_n)| =1$. 
    Recall Case 2-3 of Lemma~\ref{lem:transition}, then 
    \begin{equation*}
      \typefn(\bitseq_{n+1})
      = \typefn(\bitseq_{n} \flip{b_{n}})
      = [\tent^{n+1}(\tfrac{1}{2}), \tent^{\partner(n)+1}(\tfrac{1}{2}))
      = [\tent^{n+1}(\tfrac{1}{2}), \tent^{\partner(n+1)}(\tfrac{1}{2}))
    \end{equation*}
   holds where we used $\partner(n+1)=\partner(n)+1$ by \eqref{eq:partner_2}.
  We obtain the claim in the case. 
  \end{itemize}
 Then, we obtain \eqref{eq:elements_of_type_b}. 
\end{proof}

Then, we sequentially prove three technical lemmas, Lemmas~\ref{lemm:bits_from_partner}--\ref{lemm:back_transition}. 
\begin{lemma}\label{lemm:bits_from_partner}
  Suppose for  $\mu \in (1, 2)$ that $\tent^{i}(\frac{1}{2}) \neq \frac{1}{2}$ holds for any $i = 1, \dots, n-1$. 
    Then,
    $c_{n} = \flip{c_{\partner(n)}}$
  holds. 
\end{lemma}

\begin{proof}
  The proof is an  induction on $n$. 
  Notice that $\enc^{2}(\frac{1}{2}) = 10$, 
    meaning that $c_1=1$ and $c_2=0$. 
  Recall $\partner(2)=1$. 
  Thus, $c_2=0$ and $c_{\partner(2)} = c_1 = 1$, and 
   we obtain the claim $c_{n} = \flip{c_{\partner(n)}}$ for $n=2$.

  Inductively assuming $c_{n} = \flip{c_{\partner(n)}}$ holds for $n$, we prove  $c_{n+1} = \flip{c_{\partner(n+1)}}$. 
  Recall 
  \begin{equation*}
    \typei_n=
    \begin{cases}
      [\tent^{n}(\frac{1}{2}), \tent^{\partner(n)}(\frac{1}{2}))
      &\mbox{if $c_{n} = 0$,}\\
      (\tent^{\partner(n)}(\frac{1}{2}), \tent^{n}(\frac{1}{2})]
      &\mbox{otherwise}
    \end{cases}
    \tag{\ref{eq:elements_of_type_b}}
  \end{equation*}
  by Lemma~\ref{lemm:elements_of_type_2}. We consider three cases. 
\vspace{-1ex}
\begin{itemize}\setlength{\itemindent}{1em}\setlength{\parskip}{0.5ex}\setlength{\itemsep}{0cm} 
  \item[Case 1.]
    Suppose $\tent^{n}(\frac{1}{2}) < \frac{1}{2} < \tent^{\partner(n)}(\frac{1}{2})$
    or $\tent^{\partner(n)}(\frac{1}{2}) < \frac{1}{2} < \tent^{n}(\frac{1}{2})$.
    Note that $|\outn(\typei_n)|=2$ in the case, and hence 
    $c_{n+1} = 0$ (cf. proof of Lemma \ref{lemm:elements_of_type_2}).
    On the other hand,
     $c_{\partner(n+1)} = c_1 = 1$ since $\partner(n+1) = 1$ when $|\outn(\typei_n)|=2$  by definition \eqref{eq:partner_2}. 
    We obtain the claim $c_{n+1}=\flip{c_{\partner(n+1)}}$ in the case.
 \item[Case 2.]
    Suppose $\tent^{n}(\frac{1}{2}) < \frac{1}{2}$ and $\tent^{\partner(n)}(\frac{1}{2}) < \frac{1}{2}$. 
    They respectively imply 
     $c_{n}=c_{n+1}$ and $c_{\partner(n)}=c_{\partner(n)+1}$ by definition \eqref{def:encode1}.
   Note that $|\outn(\typei_n)|=1$ in the case, by \eqref{eq:elements_of_type_b}. 
   Thus $\partner(n+1) = \partner(n) + 1$ by definition \eqref{eq:partner_2}, 
     meaning that  $c_{\partner(n)+1}$ is exactly $c_{\partner(n+1)}$ itself.  
  The inductive assumption $c_{n}=\flip{c_{\partner(n)}}$ implies the claim $c_{n+1}=\flip{c_{\partner(n+1)}}$ in the case. 
 \item[Case 3.]
   Suppose $\tent^{n}(\frac{1}{2}) \geq \frac{1}{2}$ and $\tent^{\partner(n)}(\frac{1}{2}) \geq \frac{1}{2}$. 
   They respectively imply 
    $c_{n}=\flip{c_{n+1}}$ and 
    $c_{\partner(n)}=\flip{c_{\partner(n)+1}}$ by definition \eqref{def:encode1}.
   Note that $|\outn(\typei_n)|=1$ in the case, and hence
    $\partner(n+1) = \partner(n) + 1$ holds. 
  The inductive assumption $c_{n}=\flip{c_{\partner(n)}}$ implies the claim $c_{n+1}=\flip{c_{\partner(n+1)}}$ in the case. 
\end{itemize}
\end{proof}

\begin{lemma}\label{lemm:outdeg_from_bits}
  Suppose for  $\mu \in (1, 2)$ that $\tent^{i}(\frac{1}{2}) \neq \frac{1}{2}$ holds for any $i = 1, \dots, n-1$. 
 If $|\outn(\typei_n)|=2$ for $n$ ($n \geq 2$) 
  then $c_{\partner(n)+1} = 0$.
\end{lemma}

\begin{proof}
 We consider two cases $c_n=0$ or $1$. 
 Firstly, suppose $c_{n} = 0$. 
  Then, 
   $ \typei_n =  [\tent^{n}(\frac{1}{2}), \tent^{\partner(n)}(\frac{1}{2}))$
 by \eqref{eq:elements_of_type_b}.
 The hypothesis $|\outn(\typei_n)|=2$ implies 
  $\tent^{n}(\frac{1}{2}) < \frac{1}{2} < \tent^{\partner(n)}(\frac{1}{2})$ (recall Case 1-1 in Lemma~\ref{lem:transition}). 
 The former inequality implies $c_{n+1} = c_{n}$ while the latter inequality implies 
  $c_{\partner(n)+1} = \flip{c_{\partner(n)}}$ by definition \eqref{def:encode1}. 
  Since $c_{n} = \flip{c_{\partner(n)}}$ by Lemma \ref{lemm:bits_from_partner}, 
   we obtain $c_{n+1} = c_{\partner(n)+1} = 0$ in the case of $c_n=0$. 

 Secondly,  suppose $c_{n} = 1$. 
 Then, $ \typei_n =  (\tent^{\partner(n)}(\frac{1}{2}), \tent^{n}(\frac{1}{2})]$
  by \eqref{eq:elements_of_type_b}.
 The hypothesis $|\outn(\typei_n)|=2$ implies 
  $\tent^{\partner(n)}(\frac{1}{2}) < \frac{1}{2} < \tent^{n}(\frac{1}{2})$ (recall Case 2-1 in Lemma~\ref{lem:transition}). 
 The former inequality implies $c_{\partner(n)+1} = c_{\partner(n)}$ while the latter inequality implies 
  $c_{n+1} = \flip{c_{n}}$ by definition \eqref{def:encode1}. 
  Since $\flip{c_{n}} = c_{\partner(n)}$ by Lemma \ref{lemm:bits_from_partner}, 
 we obtain $c_{n+1} = c_{\partner(n)+1} = 0$ in the case of $c_n=1$. 
\end{proof}

\begin{lemma}\label{lemm:back_transition}
  Suppose for  $\mu \in (1, 2)$ that $\tent^{i}(\frac{1}{2}) \neq \frac{1}{2}$ holds for any $i = 1, \dots, n-1$. 
  Suppose $n$ ($n \geq 2$) satisfies $|N(\typei_{n})| = 2$.  
  Then, the following (i) and (ii) hold: 
\begin{itemize}
\item[(i)]
 $\partner(\partner(n)+1) = 1$.
\item[(ii)] $\delta(\typei_{n},1) = \flip{\typei}_{\partner(n)+1}$. 
\end{itemize}  
\end{lemma}

\begin{proof}
 As a preliminary step, we claim that the hypothesis $|N(\typei_{n})| = 2$ implies 
  \begin{equation}
   \typefn(\bitseqc_{n} 1) =  (\tent^{\partner(n)+1}(\tfrac{1}{2}), \tent(\tfrac{1}{2})] 
    \label{eq:back_transition_4}
  \end{equation}
  holds.   
  Note that $\delta(\typei_n,0) = \typei_{n+1}$, 
   which  we have proved in Cases 1-1 and 2-1 in Lemma~\ref{lemm:elements_of_type_2}. 
 The proof of \eqref{eq:back_transition_4} is similar to them. 
 We consider the cases $c_n=0$ or $1$.
  In case of $c_{n} = 0$, 
   $\typefn(\bitseqc_{n}) = [\tent^{n}(\tfrac{1}{2}), \tent^{\partner(n)}(\tfrac{1}{2}))$ by Lemma~\ref{lemm:elements_of_type_2}. 
  Recall Case 1-1-2 in Lemma~\ref{lem:transition}, then
    \begin{equation*}
     \typefn(\bitseqc_{n} 1)
      = (\tent^{\partner(n)+1}(\tfrac{1}{2}),\tent^{1}(\tfrac{1}{2})]
    \end{equation*}
   holds. 
  In case of $c_{n} = 1$, 
   $\typefn(\bitseqc_{n}) = (\tent^{\partner(n)}(\tfrac{1}{2}),\tent^{n}(\tfrac{1}{2})]$ by Lemma~\ref{lemm:elements_of_type_2}. 
  Recall Case 2-1-1 in Lemma~\ref{lem:transition}, then
    \begin{equation*}
     \typefn(\bitseqc_{n} 1)
      = (\tent^{\partner(n)+1}(\tfrac{1}{2}),\tent^{1}(\tfrac{1}{2})]
    \end{equation*}
   holds. 
   Thus we obtain \eqref{eq:back_transition_4}.

 Now we prove claim (i). 
 Notice that 
   $\typefn(\bitseqc_{n} 1) \in \typeset(n)$ by Theorem~\ref{theo:number_of_type}.
It implies that 
  there exists $k \in \{1,2\ldots,n\}$ such that 
   $\typefn(\bitseqc_{n} 1) \in \{\typei_k,\flip{\typei}_k\}$ holds. 
 The proof consists of three steps. 
 Firstly, we claim that $k \neq 1$. 
In fact, 
 $\flip{\typei}_{1} = [0, \tent(\frac{1}{2}))$ and $\typei_{1} = (0, \tent(\frac{1}{2})]$, accordingly 
     $\typefn(\bitseqc_{n} 1) = (\tent^{\partner(n)+1}(1/2), \tent(1/2)]$ cannot be one of them.  
 Secondly, we claim $\partner(k) = 1$. 
 Notice that one end of the $\typefn(\bitseqc_{n} 1) $ is $\tent(\tfrac{1}{2})$ by \eqref{eq:back_transition_4}. 
 By Lemma~\ref{lemm:elements_of_type_2}, 
   both ends of $\typei_k$ are $\tent^k(\tfrac{1}{2})$ and $\tent^{\partner(k)}(\tfrac{1}{2})$,  $\flip{\typei}_k$ as well. 
 The hypothesis requires $\tent^{k}(\frac{1}{2}) \neq \tent(\frac{1}{2})$ since $k \neq 1$ as we proved above. 
 Thus, $\tent^{\partner(k)}(\frac{1}{2}) = \tent(\frac{1}{2})$ must hold, 
   where the hypothesis again implies $\partner(k) = 1$. 
 Thirdly, we claim $k = \partner(n)+1$. 
 Now we know one end of  $\typefn(\bitseqc_{n} 1) $ is $\tent(\tfrac{1}{2})=\tent^{\partner(k)}(\frac{1}{2})$. 
 Thus the other end must satisfy 
    $\tent^{\partner(n)+1}(\frac{1}{2}) = \tent^k(\frac{1}{2})$. 
 The hypothesis allows only $k=\partner(n)+1$. 
 Now we got $\partner(k) = 1$ and $k = \partner(n)+1$, which implies (i).

 We prove (ii). 
 Note that $c_{\partner(n)+1} = 0$ since $|N(\typei_{n})| = 2$, by Lemma \ref{lemm:outdeg_from_bits}.  
 Thus, $\flip{c_{\partner(n)+1}} = 1$. 
 Then, 
  \begin{align*}
    \flip{\typei}_{\partner(n)+1}
    &= \typefn(\flip{\bitseqc_{\partner(n)+1}}) && (\mbox{by definition \eqref{def:typei}}) \nonumber \\
    &= \typefn(\flip{\bitseqc_{\partner(n)}}\,1) && (\mbox{since $\flip{c_{\partner(n)+1}} = 1$}) \nonumber \\
    &= (\tent^{\partner(n)+1}(\tfrac{1}{2}), \tent^{\partner(\partner(n)+1)}(\tfrac{1}{2})] && (\mbox{by \eqref{eq:elements_of_type_b_2}}) \nonumber \\
    &= (\tent^{\partner(n)+1}(\tfrac{1}{2}), \tent(\tfrac{1}{2})] && (\mbox{since $\partner(\partner(n)+1) = 1$, by (i)}) \nonumber \\
    &= \typefn(\bitseqc_{n} 1) && (\mbox{by \eqref{eq:back_transition_4}}) \\
    &= \delta(\typei_{n}, 1) && (\mbox{recall $\typei_{n}=\typefn(\bitseqc_{n})$})
  \end{align*}
  holds, and we obtain the claim. 
\end{proof}

 As a consequence of Lemma \ref{lemm:back_transition}, we obtain the following lemma. 

\begin{lemma}\label{lemm:back_half}
  Suppose for  $\mu \in (1, 2)$ that $\tent^{i}(\frac{1}{2}) \neq \frac{1}{2}$ holds for any $i = 1, \dots, n-1$. 
  If $|\outn(\typei_n)|=2$ holds for $n$ ($n \geq 2$) 
  then $\partner(n) \le \frac{n}{2}$. 
\end{lemma}

\begin{proof}
 The hypothesis $|\outn(\typei_n)|=2$ implies $\partner(\partner(n)+1) = 1$ by Lemma \ref{lemm:back_transition}. 
 It implies $|\outn(\typei_{\partner(n)})|=2$ by the recurrence relation \eqref{eq:partner_2}.  
 Notice that 
  $|\outn(\typei_{n-k})|=1$ holds for $k=1,2,\ldots,\partner(n)$ by definition of $\partner(n)$. 
 Thus, $\partner(n) < n - \partner(n) + 1$ must hold. 
 Since $\partner(n)$ and $n$ are integer,
  we obtain $\partner(n) \le \frac{n}{2}$. 
\end{proof}

Now, Lemma \ref{lemm:transition_function2} is easy. 
\begin{proof}[Proof of Lemma~\ref{lemm:transition_function2}]
By the hypotheses, 
  \begin{align*}
    \delta(\typei_{n}, 1)
    &= \flip{\typei}_{\partner(n)+1} && (\mbox{by Lemma~\ref{lemm:back_transition}}) \\
    &\in \left\{\flip{\typei}_{k+1} \mid 1 \le k \le \tfrac{n}{2} \right\} && (\mbox{by Lemma~\ref{lemm:back_half}}) 
  \end{align*}
holds. 
\end{proof}

  Finally, we prove Proposition~\ref{prop:const}, here. 
\begin{proposition}[Proposition~\ref{prop:const}]
 Let $\mu \in (1,2)$. 
 If there exists $n$ ($n=1,2,\ldots$) such that $\tent_{\mu}^n(\frac{1}{2}) = \frac{1}{2}$ holds 
   then $|\typeset_{n'}| \leq 2n$ holds for any $n' \geq n$. 
\end{proposition}
\begin{proof}
Suppose $\tent^n(\frac{1}{2}) = \frac{1}{2}$. 
Let $\bitseqc_{n+1} = c_1\cdots c_{n+1} =\enc^{n+1}(\tfrac{1}{2})$. 
Firstly, we claim 
$\tilde{\tent}(\se_{n+1}(\bitseqc_{n+1})) = \se(c_2\cdots c_{n+1})$. 
Let $\se_{n+1}(\bitseqc_{n+1}) = [\frac{1}{2},u)$ according to Lemma~\ref{lem:c=d}. 
We know
$\tilde{\tent}(\se_{n+1}(\bitseqc_{n+1})) = [\tilde{\tent}(\frac{1}{2}),\tilde{\tent}(u)) \subseteq \se_n(c_2\cdots c_{n+1})$ 
by Lemma~\ref{lem:section-shift}. 
Assume for a contradiction, 
 let $x \in \se(c_2\cdots c_{n+1}) \setminus [\tilde{\tent}(\frac{1}{2}),\tilde{\tent}(u))$. 
Let $x' = \tilde{\tent}(\frac{1}{2})$ and let $x'' \in (\tilde{\tent}(\frac{1}{2}),\tilde{\tent}(u))$, 
 then $x < x' <x''$ hold. 
We claim both $x'_n \geq x_n$ and $x'_n \geq x''_n$ hold 
for $x_n = \tent^n(x)$, $x'_n = \tent^n(x')$ and $x''_n = \tent^n(x'')$.  
In fact, $x'_n = \tent^n(x') =\tent^n(\tilde{\tent}(\frac{1}{2})) = \tent^{n+1}(x) = \tent(\tent^n(\frac{1}{2})) = \tent(\frac{1}{2})$ holds. 
By definition \eqref{eq:tentmap}, $\tent(\frac{1}{2}) \geq \tent(y)$ for any $y \in [0,1]$. 
Thus, we obtain $x'_n \geq x_n$ and $x'_n \geq x''_n$. 
This contradicts to Lemma~\ref{lem:encode1}, claiming either $x_n < x'_n <x''_n$ or  $x_n > x'_n > x''_n$. 
We obtain $\tilde{\tent}(\se_{n+1}(\bitseqc_{n+1})) = \se_n(c_2\cdots c_{n+1})$. 

Now, it is easy to see 
 $\typefn(\bitseqc_{n+1}) 
  = \tent^{n+1}(\se_{n+1}(\bitseqc_{n+1})) 
  = \tent^n(\tilde{\tent}(\se_{n+1}(\bitseqc_{n+1})))
  = \tent^n(\se_n(c_2\cdots c_{n+1})) 
  = \typefn(c_2\cdots c_{n+1})$. 
We obtain the claim by Lemma~\ref{lem:number_of_type3}. 
\end{proof}


\section{Analysis for Real $\mu$}\label{sec:real}
 We assumed $\mu$ rational in Section~\ref{sec:rational}, 
  to avoid some bothering arguments on Turing computability of real numbers, 
  but it is not essential. 
 This section shows that $K$ is $\Order(\log n \log \log n)$ with high probability, in expectation as well, 
  even for real $\mu$. 
 To be precise, we prove the following proposition. 
\begin{proposition}\label{prop:real-prob}
Let $\mu \in (1,2)$ be an arbitrary real, and let $c \geq 1$ be a constant.  
For convenience, 
let  
 $l_* = \max\{8c  \log_{\mu} n \log_2 \log_{\mu} n, -2\log_{\mu} \frac{\mu-1}{2}\} $.\footnote{
  Notice that $-2\log_{\mu} \frac{\mu-1}{2}>0$ for $\mu \in (1,2)$, 
  where its value is asymptotic to $\infty$ as $\mu \to 1+0$. 
  The term is negligible if $\mu -1 \gg 0$, 
  e.g., $-2\log_{\mu} \frac{\mu-1}{2}\leq 10$ if $\mu \geq \sqrt{2}$. 
  } 
Then, 
\begin{align*}
 \Pr[\requiredstgsize \ge 2 l_*]  \leq n^{-(c-1)} 
\end{align*} 
 holds. 
\end{proposition}
 Remark that it is possible to establish $\poly\log n$ average space complexity even for some real $\mu$ 
   from Proposition~\ref{prop:real-prob}, 
  using some standard (but bothering) arguments on computations with real numbers, 
   e.g., symbolic treatment of $\sqrt{2}$, ${\rm e}$, $\pi$, etc. 
 Here, we just prove Proposition~\ref{prop:real-prob}, and 
  omit the arguments on $\poly\log n$ average space for real $\mu$.

 The proof of Proposition~\ref{prop:real-prob} is similar to Lemma~\ref{lem:l*}, 
  but we have to prove the following lemma  
   without the assumption of $\mu$ being rational, in contrast to Lemma~\ref{lemm:existence_of_short_q_s}. 
\begin{lemma}
  \label{lemm:existence_of_short_q_s-r}
 Let $\mu \in (1,2)$ be an arbitrary real, and let $c \geq 1$ be a constant.  
 For convenience, 
 let $l_* = \max\{8c  \log_{\mu} n \log_2 \log_{\mu} n, -2\log_{\mu} \frac{\mu-1}{2}\} $. 
 Then, there exists $l$ such that  $l \leq l_*$ and 
  \begin{equation}
    \label{eq:shortqs_0-r}
    \frac{|\typei_{2l}|}{\mu^l |\typei_l|} \le n^{-c}
  \end{equation}
  hold. 
\end{lemma}

\begin{proof}
 For convenience, let 
   $l_i =  2^i $ for $i = 1, 2, \ldots$. 
  Assume for a contradiction that 
     \eqref{eq:shortqs_0-r} never hold for any $l_1,l_2,\ldots,l_k$, where let 
\begin{align}
 l_k \geq  \max \left\{4c  \log_{\mu} n \log_2 \log_{\mu} n, -\log_{\mu} \frac{\mu-1}{2} \right\} 
 \quad\left( = \tfrac{l_*}{2} \right)
\end{align} 
  hold. Notice that such $l_k$ exists at most $l_*$. 
 In other words, 
  \begin{equation}
    |\typei_{l_{i+1}}| > n^{-c} \mu^{l_{i}} |\typei_{l_{i}}|
  \end{equation}
  holds every $i=1,2,\ldots,k$. 
 Thus, we inductively obtain that 
  \begin{align}
    |\typei_{l_{k+1}}|
    &> n^{-c} \mu^{l_{k}} |\typei_{l_{k}}| \nonumber \\
    &> n^{-c} \mu^{l_{k}} \mu^{l_{k-1}} |\typei_{l_{k-1}}| \nonumber \\
    &> \dots \nonumber \\
    &> n^{-ck} \mu^{l_{k}} \mu^{l_{k-1}} \dots \mu^{l_1} |\typei_{l_1}|
    \label{eq:shortqs_2-r}
  \end{align}
   holds. 
Note that 
  \begin{align}
    \label{eq:shortqs_7-r}
    |\typei_{l_1}|
    \ =\ |\typei_2|
    \ =\ \frac{\mu}{2} - \mu\left(1-\frac{\mu}{2}\right)
    \ =\  \mu\frac{\mu-1}{2} 
    \ >\ 0
  \end{align}
 holds.  
 Then,  \eqref{eq:shortqs_2-r} and \eqref{eq:shortqs_7-r} imply 
  \begin{align}
    |\typei_{l_{k+1}}|
    &>  n^{-c k}  \cdotp  \mu^{2^{k} + 2^{k-1} + \dots + 2^1} \cdotp \mu\frac{\mu-1}{2} 
    \label{eq:shortqs_6-r}
  \end{align}
 holds. 
  By taking the $\log_{\mu}$ of the both sides of \eqref{eq:shortqs_6-r}, 
    we see that 
  \begin{align}
    \log_{\mu}{|\typei_{l_{k+1}}|}
    &> -c k \log_{\mu}n + 2^{k+1} - 1 + \log_{\mu}\frac{\mu-1}{2} \nonumber \\
    &= (2^{k} -ck \log_{\mu}n - 1) + \left(2^{k} + \log_{\mu}\frac{\mu-1}{2}\right) 
    \label{eq:shortqs_8-r}
  \end{align}
   holds. 
 Notice that $f(x) =  2^x - ax - 1 $  ($a > 0$) is monotone increasing for $x \geq 1+\log_2 a$. 
 Since $l_k = 2^k \geq 4c  \log_{\mu} n \log_2 \log_{\mu} n $, 
  \begin{align}
    2^{k} -c k \log_{\mu}n - 1 
    &\geq 4c  \log_{\mu} n \log_2 \log_{\mu} n - c \log_{2} (4c \log_{\mu} n \log_2 \log_{\mu} n ) \log_{\mu} n -1 \nonumber\\
    &= c \log_{\mu} n ( 4 (\log_2 \log_{\mu} n) - \log_{2} (4c \log_{\mu} n \log_2 \log_{\mu}  n) ) -1 \nonumber\\
    &= c \log_{\mu} n ( 4 (\log_2 \log_{\mu} n) - \log_2 4c -\log_2 \log_{\mu} n -\log_{2} \log_{2} \log_{\mu} n ) ) -1 \nonumber\\
    &= c \log_{\mu} n ( 3 (\log_2 \log_{\mu} n) - \log_2 4c -\log_{2} \log_{2} \log_{\mu} n ) ) -1 
    \label{eq:shortqs_1-r}
  \end{align}
  holds.
  It is not difficult to see that 
\begin{align*} 
&  \log_2 \log_{\mu} n - \log_2 4c > 0  \qquad\qquad \mbox{and}\\
& \log_2 \log_{\mu} n-\log_{2} \log_{2} \log_{\mu}  n >0
\end{align*}
 hold for sufficiently large $n$, and hence 
\begin{align}
 \eqref{eq:shortqs_1-r} 
  & \geq c \log_{\mu} n \log_2 \log_{\mu} n  -1 \nonumber \\
  & \geq 0
\label{eq:shortqs_11-r}
\end{align}
 holds for sufficiently large $n$.

 Since  $2^k \geq  -\log_{\mu} \frac{\mu-1}{2}$ by definition,  
 it is also not difficult to observe that 
  \begin{align}
  2^{k} + \log_{\mu}\frac{\mu-1}{2}
   \ \geq \  -\log_{\mu}\frac{\mu-1}{2} + \log_{\mu}\frac{\mu-1}{2}
   \ =\ 0
    \label{eq:shortqs_9-r}
  \end{align}
  holds.  
  Equations \eqref{eq:shortqs_8-r}, \eqref{eq:shortqs_11-r} and \eqref{eq:shortqs_9-r} imply that $\log_{n}{|\typei_{l_{k+1}}|} > 0$, 
     meaning that $|\typei_{l_{k+1}}| > 1$. 
  At the same time, 
    notice that any segment-type $\typej$ satisfies $\typej \subseteq [0,1]$,  
     meaning that 
    $|\typei_{l_{k+1}}| \le 1$. 
  Contradiction. 
\end{proof}

Now, the rest of the proof of Proposition~\ref{prop:real-prob} is essentially the same as  Lemma~\ref{lem:l*}. 
\begin{proof}[Proof of Proposition~\ref{prop:real-prob}]
 Lemma~\ref{lemm:existence_of_short_q_s-r} implies that there exists $l$ such that 
  $l \leq l_* $ and 
\begin{align}
 \frac{|\typei_{2l}|}{\mu^{l} |\typei_{l}|} \leq n^{-c}
\label{eq:l*1-r}
\end{align}
holds.  
 Let $A_t$ ($t=1,\ldots,n$) denote the event that $Z_t$ reaches the level $2 l$ {\em for the first time}. 
 Then, 
\begin{align*}
 \Pr[\requiredstgsize \ge 2l] 
 &= \Pr\left[\bigvee_{t=0}^{n} A_t \right] \\
 &\le \sum_{t=0}^n \Pr\left[A_t\right] \\
  &\leq n \Pr[ \lev(Z_t) = 2l  \mid \lev(Z_{t-l}) = l ] \\
  &\leq n \frac{|\typei_{2l}|}{\mu^{l} |\typei_{l}|} \\
  & \leq n^{-(c-1)} 
   && (\mbox{by \eqref{eq:l*1-r}}) 
\end{align*}
  holds. 
We remark that 
$\Pr[\requiredstgsize \ge 2l_*] 
 \leq 
 \Pr[\requiredstgsize \ge 2l] $
is trivial since $l < l_*$. 
\end{proof}

\begin{corollary}\label{cor:exp-real}
 Let $\mu \in (1,2)$ be an arbitrary real. 
 Let $l'_* = \max\{16  \log_{\mu} n \log_2 \log_{\mu} n, -2\log_{\mu} \frac{\mu-1}{2}\} $,
 for convenience. Then,  
 $\E[K] \leq 2l'_*$. 
\end{corollary}

\begin{proof}
 Proposition~\ref{prop:real-prob} with $c=2$ implies 
\begin{align}
 \Pr[\requiredstgsize \ge 2l'_*] \leq n^{-1}
\label{eq:exp-real1}
\end{align}
 holds. Then 
\begin{align*}
  \E[\requiredstgsize]
   &= \sum_{k=1}^n k\Pr[\requiredstgsize =k]  \nonumber \\
   &= \sum_{k=1}^{2l'_*-1} k \Pr[\requiredstgsize =k]  + \sum_{k=2l'_*}^n k \Pr[\requiredstgsize =k]  \nonumber \\
   &\leq (2l'_* - 1) \Pr[\requiredstgsize \leq 2l'_*-1] + n \Pr[\requiredstgsize \ge 2l'_*] \nonumber \\
   &\leq (2l'_* - 1) + n \Pr[\requiredstgsize \ge 2l'_*] \nonumber \\
   &\leq (2l'_* - 1) + 1
    && (\mbox{by \eqref{eq:exp-real1}}) \\
   &= 2l'_*
  \end{align*}
holds. 
\end{proof}

\section{Concluding Remark}
 This paper showed that $B \sim {\cal D}_n$ is realized in $\Order(\log^2 n)$ space, in average (Theorem~\ref{thm:main}). 
 An extension to the smoothed analysis, beyond average, is a near future work. 
 Another future work is an extension to the baker's map, 
   which is a chaotic map of piecewise but 2-dimensional. 
 For the purpose, we need an appropriately extended notion of the segment-type. 
 Another future work is an extension to the logistic map, 
  which is a chaotic map of 1-dimensional but quadratic. 
 Some techniques of random number transformation may be available for it. 
 The time complexity is another interesting topic 
   to decide $b_n\in \{0,1\}$ as given a rational $x=p/q$ for a fixed $\mu \in \mathbb{Q}$. 
 Is it possible to compute in time polynomial in the input size $\log p+\log q+\log n$? 
 It might be NP-hard, but we could not find a result. 

\section*{Acknowledgement}
 The authors are grateful to Masato Tsujii for the invaluable comments, particularly about the Markov extension. 
 The authors would like to thank Yutaka Jitsumatsu, Katsutoshi Shinohara and Yusaku Tomita 
  for the invaluable discussions, that inspire this work. 
 The authors are deeply grateful to Jun'ichi Takeuchi for his kind support.  
 This work is partly supported by JSPS KAKENHI Grant Numbers JP21H03396. 

\bibliographystyle{plain}

\appendix
\section{Proofs Remaining from Section~\ref{sec:tent-map}}\label{apx:tent-map}
\begin{proposition}[Proposition~\ref{prop:tent-expansion}]
 Suppose $\enc^{\infty}(x) = b_1b_2\cdots$ for $x \in [0,1)$. 
 Then, $(\mu - 1) \sum_{i=1}^{\infty} b_{i} \mu^{-i} = x$. 
\end{proposition}

\begin{proof}
  Let 
  \begin{equation}
    \label{eq:proof_of_tent_expansion_1}
    a_{i} = b_{i} \oplus b_{i-1}
  \end{equation}
  for $i \ge 1$, where $b_{0} = 0$ for convenience. 
By \eqref{def:encode2}, 
 $a_i=0$ if and only if $x_i<1/2$. 
  Then, \eqref{eq:tentmap} implies  
  \begin{equation}
    x_{i} = (-1)^{a_{i}} \mu x_{i-1} + \mu a_{i}
\label{eq20230615a}
  \end{equation}
  holds for any $i$.
 Recursively applying \eqref{eq20230615a}, we see 
  \begin{equation*}
    x_{n}
    = (-1)^{\bigoplus_{i=1}^n a_i} \mu^{n} x_{0} + (-1)^{\bigoplus_{i=2}^n a_i} \mu^{n} a_{1} + (-1)^{\bigoplus_{i=3}^n a_i} \mu^{n-1} a_{2} + \dots + (-1)^{0} \mu^{1} a_{n}
  \end{equation*}
   holds for any $n$. 
  By Eq. \eqref{eq:proof_of_tent_expansion_1},
  \begin{equation*}
    x_{n}
    = (-1)^{b_{n}} \mu^{n} x_{0} + (-1)^{b_{n} \oplus b_{1}} \mu^{n} (b_{1} \oplus 0) + (-1)^{b_{n} \oplus b_{2}} \mu^{n-1} (b_{2} \oplus b_{1}) + \dots + (-1)^{b_{n} \oplus b_{n}} \mu^{1} (b_{n} \oplus b_{n-1}).
  \end{equation*}
  Multiplying $(-1)^{b_{n}} \mu^{-n}$ in both sides,
  \begin{align*}
    (-1)^{b_{n}} \mu^{-n} x_{n}
    &= x_{0} + (-1)^{b_{1}} (b_{1} \oplus 0) + (-1)^{b_{2}} \mu^{-1} (b_{2} \oplus b_{1}) + \dots + (-1)^{b_{n}} \mu^{-(n-1)} (b_{n} \oplus b_{n-1}) \\
    &= x_{0} - b_{1} - \mu^{-1} (b_{2} - b_{1}) + \dots - \mu^{-(n-1)} (b_{n} - b_{n-1}) \\
    &= x_{0} - (\mu - 1)(\mu^{-1} b_{1} + \mu^{-2} b_{2} + \cdots \mu^{-n} b_{n}) \\
    &= x - (\mu - 1) \sum_{i=1}^{n} b_{i} \mu^{-i}.
  \end{align*}
  Notice that $x_{n} \in [0, 1)$ holds for any $n$ and $\mu$.
  Thus,
  \begin{equation}
    \left|x - (\mu - 1) \sum_{i=1}^{n} b_{i} \mu^{-i} \right| \le \mu^{-n}
    \label{eq:error}
  \end{equation}
  holds for any $n$. 
  Clearly $\mu^{-n} \to 0$ asymptotic to $n \to \infty$. 
  We obtain the claim. 
\end{proof}

As a preliminary step of Propositions~\ref{prop:order} and \ref{prop:-heikai}, we prove Lemma~\ref{lem:encode1} first. 
\begin{lemma}[Lemma~\ref{lem:encode1}]
  Let $x,x' \in [0,1)$ satisfy $x < x'$. 
  Let $\enc(x) =b_1b_2\cdots $ and  $\enc(x^{\prime})=b'_1b'_2\cdots $. 
 If $b_i = b'_i$ hold for all $i=1,\ldots,n$ then 
  \begin{align}
  \begin{cases}
    x_n < x^{\prime}_n  & \mbox{if $b_n = b^{\prime}_n = 0$, }\\ 
    x_n > x^{\prime}_n & \mbox{if $b_n = b^{\prime}_n = 1$}
 \end{cases}
  \tag{\ref{eq:tentexpansion_lexicography_1}}
    \end{align}
  holds. 
\end{lemma}
\begin{proof}
  We prove it by an induction on $i =1,2,\ldots,n$. 
  Consider the case of $i=1$. 
  By \eqref{def:encode0}, $b_1 = b^{\prime}_1 = 0$ only when $x ,x' < 1/2$, 
    accordingly $x_1 = \mu x < \mu x' =x'_1$ holds by \eqref{eq:tentmap}. 
  Similarly,  
   if $b_1 = b^{\prime}_1 = 1$ then $x ,x' \geq 1/2$ by \eqref{def:encode0}, 
    accordingly $x_1 = \mu (1-x) \geq \mu (1-x') =x'_1$ holds by \eqref{eq:tentmap}. 
    We obtain \eqref{eq:tentexpansion_lexicography_1} for $i=1$. 
  
  Inductively assuming \eqref{eq:tentexpansion_lexicography_1} for $i$ ($\leq n-1$), we prove it for $i+1$.  
 According to \eqref{def:encode2}, we consider four cases. 
  \vspace{-1ex}
  \begin{itemize}\setlength{\itemindent}{1em}\setlength{\parskip}{0.5ex}\setlength{\itemsep}{0cm} 
    \item[Case 1.]
    Suppose $b_{i} = b^{\prime}_{i} = 0$ and $b_{i+1} = b^{\prime}_{i+1} = 0$. 
    Then, $x_{i} < x^{\prime}_{i} < \frac{1}{2}$ by \eqref{eq:tentexpansion_lexicography_1} and \eqref{def:encode2}.
    Therefore, $x_{i+1} = \mu x_{i} < \mu x^{\prime}_{i} = x^{\prime}_{i+1}$ by \eqref{eq:tentmap}. 
    \item[Case 2.]
    Suppose $b_{i} = b^{\prime}_{i} = 0$ and $b_{i+1} = b^{\prime}_{i+1} = 1$. 
    Then,  $\frac{1}{2} \le x_{i} < x^{\prime}_{i}$.
    Therefore, $x_{i+1} = \mu (1 - x_{i}) > \mu (1 - x^{\prime}_{i}) = x^{\prime}_{i+1}$.
    \item[Case 3.]
    Suppose $b_{i} = b^{\prime}_{i} = 1$ and $b_{i+1} = b^{\prime}_{i+1} = 1$. 
    Then, $x^{\prime}_{i} < x_{i} \le \frac{1}{2}$.
    Therefore, $x_{i+1} = \mu x_{i} > \mu x^{\prime}_{i} = x^{\prime}_{i+1}$.
    \item[Case 4.]
    Suppose $b_{i} = b^{\prime}_{i} = 1$ and $b_{i+1} = b^{\prime}_{i+1} = 0$. 
    Then, $\frac{1}{2} < x^{\prime}_{i} < x_{i}$.
    Therefore, $x_{i+1} = \mu (1 - x_{i}) < \mu (1 - x^{\prime}_{i}) = x^{\prime}_{i+1}$.
  \end{itemize}
 In each case we obtain  \eqref{eq:tentexpansion_lexicography_1}.  
\end{proof}

\begin{proposition}[Proposition~\ref{prop:order}]
  For any $x, x^{\prime} \in [0, 1)$,
  \begin{align*}
    x \leq x^{\prime} &\Rightarrow \enc_{\mu}^n(x) \preceq \enc_{\mu}^n(x^{\prime})  
  \end{align*}
  hold where $\preceq$ denotes the {\em lexicographic order}, 
   that is $b_{i_*}=0$ and $b'_{i_*}=1$ at $i_* = \min\{j \in \{1,2,\ldots \} \mid b_j \neq b'_j\}$
   for $\enc^n(x) =b_1b_2\cdots b_n$ and  $\enc^n(x^{\prime})=b'_1b'_2\cdots b'_n$
   unless  $\enc^n(x) = \enc^n(x')$. 
\end{proposition}
\begin{proof}
  The claim is trivial for $x=x'$. 
  Suppose $x < x'$, and let $i_* = \min\{j \in \{1,2,\ldots \} \mid b_j \neq b'_j\}$, 
    where let $i_* = n+1$ if  $\enc^n(x) =\enc^n(x^{\prime})$, for convenience. 
 By Lemma~\ref{lem:encode1}, 
  we know 
  \begin{align}
  \begin{cases}
    x_{i_*-1} < x^{\prime}_{i_*-1}  & \mbox{if $b_{i_*-1} = b^{\prime}_{i_*-1} = 0$, }\\ 
    x_{i_*-1} > x^{\prime}_{i_*-1} & \mbox{if $b_{i_*-1} = b^{\prime}_{i_*-1} = 1$}
 \end{cases}
  \label{eq:tentexpansion_lexicography_2}
  \end{align}
  holds. 
 Then, we confirm $\enc^n(x) \preceq \enc^n(x^{\prime})$. 
 Consider two cases. 
  Firstly, suppose $b_{i_*-1} = b^{\prime}_{i_*-1} = 0$. 
  The hypothesis $b_{i_*} \neq b'_{i_*}$ requires $x_{i_*-1} < 1/2 \leq x^{\prime}_{i_*-1}$. 
  Then, we obtain $b_{i_*} =0$ and $b^{\prime}_{i_*} = 1$ by \eqref{def:encode2}, in the case. 
  Next, suppose $b_{i_*-1} = b^{\prime}_{i_*-1} = 1$. 
  The hypothesis $b_{i_*} \neq b'_{i_*}$ requires $x_{i_*-1} > 1/2 \geq x^{\prime}_{i_*-1}$. 
  Then, we obtain $b_{i_*} =0$ and $b^{\prime}_{i_*} = 1$ by \eqref{def:encode2}, in the case. 
  In both cases, we obtain $\enc^n(x) \preceq \enc^n(x^{\prime})$. 
\end{proof}

\begin{proposition}[Proposition~\ref{prop:-heikai}]
The $n$-th iterated tent code is right continuous, i.e., $\enc_{\mu}^n(x) = \enc_{\mu}^n(x+0)$. 
\end{proposition}
\begin{proof}
The proof is similar to Proposition~\ref{prop:order}. 
To begin with, we remark that 
  if $\enc^n(x) \prec \enc^n(x^{\prime})$ then $x < x^{\prime}$, 
 by a contraposition of Proposition~\ref{prop:order}. 
Assume for contradiction that  $\enc^n(x) \prec \enc^n(x+0)$ holds. 
Let $\enc^n(x) =b_1b_2\cdots b_n$ and $\enc^n(x')=b'_1b'_2\cdots b'_n$ where $x'=x+\epsilon_n$ and $0 < \epsilon_n \ll \mu^{-n}$, and 
 let $i_* = \min\{j \in \{1,2,\ldots n\} \mid b_j \neq b'_j\}$. 
Recall \eqref{eq:tentexpansion_lexicography_1}, then we consider two cases. 
\vspace{-1ex}
\begin{itemize}\setlength{\itemindent}{1em}\setlength{\parskip}{0.5ex}\setlength{\itemsep}{0cm} 
\item[Case 1.]
  Suppose $b_{i_*-1} = b^{\prime}_{i_*-1} = 0$. 
  The hypothesis $b_n \neq b'_n$ requires $x_{i_*-1} < 1/2 \leq x'_{i_*-1}$  due to \eqref{def:encode2}. 
 Let $\epsilon_n \ll \mu^{-n}(\frac{1}{2}-x_{i_*-1}) $, then $x'_{i_*-1} < 1/2$ holds (recall \eqref{eq:error} in the proof of Proposition~\ref{prop:tent-expansion}). 
 Contradiction. 
\item[Case 2.]
  Suppose $b_{i_*-1} = b^{\prime}_{i_*-1} = 1$. 
  The hypothesis $b_n \neq b'_n$ requires $x_{i_*-1} > 1/2 \geq x^{\prime}_{i_*-1}$ due to~\eqref{def:encode2}.
 Let $\epsilon_n \ll \mu^{-n}(x_{i_*-1}-\frac{1}{2}) $, then $x'_{i_*-1} > 1/2$ holds. 
 Contradiction. 
\end{itemize}
 We obtain the claim. 
\end{proof}

\section{Proofs Remaining from Section~\ref{sec:algo}}
\subsection{Proofs of Lemmas \ref{lem:c=d} and \ref{lem:hanten}}\label{apx:c=d}
\begin{lemma}[Lemma~\ref{lem:c=d}]
 Let $\bitseqc_n = \enc^n(\frac{1}{2})$. 
 Then, $\bitseqc_n = \bitseqd_n$ and  $\flip{\bitseqc_n} = \bitseqd'_n$ hold, where 
  $\bitseqd_n = \min\{ \bitseq_n \in {\cal L}_n \mid b_1 = 1  \mbox{ where } \bitseq_n  = b_1 \cdots b_n \}$ and 
  $\bitseqd'_n = \max\{ \bitseq_n \in {\cal L}_n \mid b_1 = 0  \mbox{ where } \bitseq_n  = b_1 \cdots b_n \}$, 
\end{lemma}
The former claim is trivial by Proposition~\ref{prop:order} and \eqref{eq:first-bit}. 
The latter claim comes from the following fact. 
\begin{lemma}[Lemma~\ref{lem:hanten}]
$\flip{\enc^n(x)} = \enc^n(1-x)$ 
 unless $x = \min S_n(x)$. 
\end{lemma}
\begin{proof}
 Without loss of generality, we may assume that $x \geq 1/2$. 
 The claim is trivial for $n=1$ 
   since $\enc(x) = 1$ and $\enc(1-x) = 0$ by \eqref{def:encode0} unless $x = 1/2$. 
 Notice that $\se_1(1/2) = [1/2,1)$, 
  meaning that $1/2 = \min \se_1(1/2)$. 

 Inductively assuming the claim holds for $n-1$, 
   we prove it for $n \geq 2$. 
 For convenience, 
  let $\enc^n(x) = b_1\cdots b_n$ and  
  let $\enc^n(1-x) = b'_1\cdots b'_n$. 
 Then, 
  $\flip{b_i} = b'_i$ holds for $i \leq n-1$ by the inductive assumption, and 
  it is enough to prove $\flip{b_n} = b'_n$. 
 Recall that $b_n$ (resp.\ $b'_n$) is determined by $\tent^{n-1}(x)$ and $b_{n-1}$ (resp.,  $\tent^{n-1}(1-x)$ and $b'_{n-1}$) by \eqref{def:encode1}. 
 We also remark that 
  $\tent(x) = \tent(1-x)$ by \eqref{eq:tentmap}, and then 
  $\tent^{n-1}(x) = \tent^{n-1}(1-x)$ inductively. 
 By the inductive assumption,  
  $\flip{b_{n-1}} = b'_{n-1}$ holds, 
  which implies with \eqref{def:encode1} that $\flip{b_n} = b'_n$ 
  unless $\tent^{n-1}(x)=\tent^{n-1}(1-x)=1/2$. 
 We obtain $\flip{\enc^n(x)} = \enc^n(1-x)$ in the case. 

 To finalize the proof, we need two more facts. 
 If $\tent^{n-1}(x)=1/2$ then $x = \min \se_n (x)$ holds, 
  by Lemma~\ref{lem:endmin} just below. 
 The second one is that 
  \eqref{eq:subdiv-sec} implies that 
   if $x = \min \se_{n-1} (x)$ then $x = \min \se_n (x)$ holds, 
   which guarantees
    the induction hypothesis that $x \neq \min \se_{n-1} (x)$ holds for any  
   $x \neq  \min \se_n (x)$. 
 Now, we obtain the lemma. 
\end{proof}

\begin{lemma}\label{lem:endmin}
  If $\tent^n(x)=1/2$ then 
  $x = \min \se_{n+1} (x)$ for $x \in [0,1)$ and $n=1,2,\ldots$. 
\end{lemma}
\begin{proof}
It feels almost trivial by the definition of $b_i$ (cf.~\eqref{def:encode1}) and the right continuity of $\se_n(x)$ (cf.~Proposition~\ref{prop:-heikai}), but 
we directly prove it using Lemma~\ref{lem:encode1}. 

 If $x = \min \se_n(x)$ then $x = \min \se_{n+1}(x)$ by \eqref{eq:subdiv-sec}. 
 Suppose $x \neq \min \se_n(x)$ and $\tent^n(x)=1/2$. 
 Then, we prove $\enc^{n+1}(x) \neq \enc^{n+1}(x-0)$. 
 Let $\enc^{n+1}(x) = b_1 \cdots b_{n+1}$ and  $\enc^{n+1}(x-0) = b'_1,\ldots,b'_{n+1}$. 
 Notice that $b_i = b'_i$ for $i=1,\ldots,n$ since $x  \neq \min \se_n(x)$ (i.e., $x  \in {\rm int} \se_n(x)$). 
 By Lemma~\ref{lem:encode1}, 
  \begin{align}
  \begin{cases}
    \tent^n(x-0) < \tent^n(x) & \mbox{if $b_n = b'_n = 0$, }\\ 
    \tent^n(x-0) > \tent^n(x) & \mbox{if $b_n = b'_n = 1$}
 \end{cases}
  \label{eq:tentexpansion_lexicography_3}
  \end{align}
 holds. 
 In each case, 
  the assumption $\tent^n(x)=1/2$ implies 
  $b_{n+1} =1$ and $b_{n+1}=0$ according to \eqref{def:encode1}. 
 Now the claim is clear. 
\end{proof}

\begin{proof}[Proof of Lemma~\ref{lem:c=d}]
 The former claim $\bitseqc_n=\bitseqd_n$ is trivial 
  since $1/2 \in \se_n(\bitseqd_n)$ by \eqref{eq:first-bit}.  
 The latter claim $\flip{\bitseqc_n}=\bitseqd'_n$ is also easy from Lemma~\ref{lem:hanten} 
  since $1/2 \in \se_n(\bitseqd_n)$ and and the right continuity of a section (cf.~Proposition~\ref{prop:-heikai}). 
\end{proof}

\subsection{Proof of Lemma~\ref{lem:transition}}\label{apx:transition}
\begin{lemma}[Lemma~\ref{lem:transition}]
Let $x \in [0,1)$. 
 (1) Suppose $\typefn^n(x) = [v,u)$ ($v<u$). 
 We consider three cases concerning the position of $\tfrac{1}{2}$ relative to $[v,u)$. 
\vspace{-1ex}
\begin{itemize}\setlength{\itemindent}{4em}\setlength{\parskip}{0.5ex}\setlength{\itemsep}{0cm} 
\item[Case 1-1:]  $v < \frac{1}{2} < u$. 
\vspace{-1ex}
\begin{itemize}\setlength{\itemindent}{4em}\setlength{\parskip}{0.5ex}\setlength{\itemsep}{0cm} 
\item[Case 1-1-1.] If $\tent^n(x) < 1/2$ then $\typefn^{n+1}(x) = [\tent(v),\tent(\tfrac{1}{2}))$, and $b_{n+1}=0$. 
\item[Case 1-1-2.] If $\tent^n(x) \geq 1/2$ then $\typefn^{n+1}(x) =  (\tent(u),\tent(\tfrac{1}{2})]$, and $b_{n+1}=1$. 
\end{itemize}
\item[Case 1-2:] $u \leq \frac{1}{2}$. Then $\typefn^{n+1}(x) = [\tent(v),\tent(u))$, and $b_{n+1}=0$.  
\item[Case 1-3:] $v \geq \frac{1}{2}$. Then $\typefn^{n+1}(x) = (\tent(u),\tent(v)]$, and $b_{n+1}=1$. 
\end{itemize}
\vspace{-1ex}

\noindent (2) Similarly, suppose $\typefn^n(x) = (v,u]$ ($v<u$). 
\vspace{-1ex}
\begin{itemize}\setlength{\itemindent}{4em}\setlength{\parskip}{0.5ex}\setlength{\itemsep}{0cm} 
\item[Case 2-1:]  $v < \frac{1}{2} < u$. 
\vspace{-1ex}
\begin{itemize}\setlength{\itemindent}{4em}\setlength{\parskip}{0.5ex}\setlength{\itemsep}{0cm} 
\item[Case 2-1-1.] If $\tent^n(x) \leq 1/2$ then $\typefn^{n+1}(x) = (\tent(v),\tent(\tfrac{1}{2})]$, and $b_{n+1}=1$. 
\item[Case 2-1-2.] If $\tent^n(x) > 1/2$ then $\typefn^{n+1}(x) =  [\tent(u),\tent(\tfrac{1}{2}))$, and $b_{n+1}=0$. 
\end{itemize}
\item[Case 2-2:] $u \leq \frac{1}{2}$. Then $\typefn^{n+1}(x) = (\tent(v),\tent(u)]$, and $b_{n+1}=1$. 
\item[Case 2-3:] $v \geq \frac{1}{2}$. Then $\typefn^{n+1}(x) = [\tent(u),\tent(v))$, and $b_{n+1}=0$. 
\end{itemize}
\end{lemma} 
\begin{proof}
 To begin with, we recall three facts. 
 i) The segment-type is defined by $\typefn^n(x) = \{\tent^n(x) \mid x \in \se_n(\bitseq_n) \}$ by \eqref{def:type}. 
 ii) $\se_n(\bitseq_n) = \se_{n+1}(\bitseq_n0) \cup \se_{n+1}(\bitseq_n1)$ by \eqref{eq:subdiv-sec}. 
 Particularly,  $\inf\se_n(\bitseq_n) = \inf \se_{n+1}(\bitseq_n0)$ and $\sup \se_n(\bitseq_n) = \sup\se_{n+1}(\bitseq_n1)$ 
  by Proposition~\ref{prop:order}. 
 iii) $b_{n+1}$ depends on $b_n$ and $x_n$ by \eqref{def:encode2}. 
 The following proof is based on the idea similar to Lemma~\ref{lem:heikai_n}. 

 (1) Suppose $\typefn^n(x) = [v,u)$ ($v<u$), 
 i.e.,  $\tent^n$ is monotone increasing in $\se_n(x)$ by Lemma~\ref{lem:heikai_n}. 
\vspace{-1ex}
\begin{itemize}\setlength{\itemindent}{2em}\setlength{\parskip}{0.5ex}\setlength{\itemsep}{0cm} 
\item[Case 1-1:]  $v < \frac{1}{2} < u$. 
 Let $x^* \in \se_n(x)$ satisfy $\tent^n(x^*) (= x^*_n) = \frac{1}{2}$. 
 Clearly, $\se_n(\bitseq_n)$ is divided into $\se_{n+1}(\bitseq_n0)$ and $\se_{n+1}(\bitseq_n1)$ at $x^*$,  
 i.e., $\se_n(\bitseq_n)=[\inf \se_{n+1}(\bitseq_n0),x^*) \cup [x^*,\se_{n+1}(\bitseq_n1))$. 
\vspace{-1ex}
\begin{itemize}\setlength{\itemindent}{2em}\setlength{\parskip}{0.5ex}\setlength{\itemsep}{0cm} 
\item[Case 1-1-1.] 
  If $\tent^n(x) < 1/2$ then $x_n \in [v,\frac{1}{2})$. 
  Thus, $\tent^{n+1}(x) =\tent(x_n) = \mu x_n$. 
  Accordingly, 
   $\tent^{n+1}(\inf \se_{n+1}(x)) = \tent(\tent^n(\inf \se_n(x))) =\tent(v) = \mu v 
     \leq \tent(x_n)  
     < \mu \frac{1}{2} = \tent(\frac{1}{2}) =\tent(x^*_n)=\tent^{n+1}(x^*)$ since $v \leq x_n < \frac{1}{2}$ 
     (cf. the proof of  Lemma~\ref{lem:heikai_n}). 
  We obtain $\typefn^{n+1}(x) = [\tent(v),\tent(\tfrac{1}{2}))$.  $b_{n+1}=0$ is clear by Lemma~\ref{lem:heikai_n}. 
\item[Case 1-1-2.] 
  If $\tent^n(x) \geq 1/2$ then $x_n \in [\frac{1}{2},u)$. 
  Thus, $\tent^{n+1}(x) =\tent(x_n) = \mu (1-x_n)$. 
  Accordingly, $\tent(\frac{1}{2}) = \mu (1-\frac{1}{2}) \geq \tent(x_n)  > \mu (1-u) = \tent(u)$ since $\frac{1}{2} \leq x_n < u$. 
  We obtain $\typefn^{n+1}(x) = (\tent(u),\tent(\tfrac{1}{2})]$. 
\end{itemize}
\item[Case 1-2:] $u \leq \frac{1}{2}$. 
 Then, $x_n <\frac{1}{2}$. 
  Thus, $\tent^{n+1}(x) =\tent(x_n) = \mu x_n$. 
  Accordingly, $\tent(v) = \mu v \leq \tent(x_n)  < \mu u = \tent(u)$ since $v \leq x_n < u$. 
  We obtain $\typefn^{n+1}(x) = [\tent(v),\tent(u))$. 
\item[Case 1-3:] $v \geq \frac{1}{2}$. 
 Then, $x_n \geq \frac{1}{2}$. 
  Thus, $\tent^{n+1}(x) =\tent(x_n) = \mu (1-x_n)$. 
  Accordingly, $\tent(v) = \mu (1-v) \geq \tent(x_n)  > \mu (1-u) = \tent(u)$ since $v \leq x_n < u$. 
 We obtain $\typefn^{n+1}(x) = (\tent(u),\tent(v)]$. 
\end{itemize}
\vspace{-1ex}

\noindent (2) Suppose $\typefn^n(x) = (v,u]$ ($v<u$), 
i.e.,  $\tent^n$ is monotone decreasing. 
\vspace{-1ex}
\begin{itemize}\setlength{\itemindent}{2em}\setlength{\parskip}{0.5ex}\setlength{\itemsep}{0cm} 
\item[Case 2-1:]  $v < \frac{1}{2} < u$. 
\vspace{-1ex}
\begin{itemize}\setlength{\itemindent}{2em}\setlength{\parskip}{0.5ex}\setlength{\itemsep}{0cm} 
\item[Case 2-1-1.] 
  If $\tent^n(x) \leq 1/2$ then $x_n \in (l,\frac{1}{2}]$. 
  Thus, $\tent^{n+1}(x) =\tent(x_n) = \mu x_n$. 
  Accordingly, $\tent(v) = \mu v < \tent(x_n)  \leq \mu \frac{1}{2} = \tent(\frac{1}{2})$ since $v \leq x_n < \frac{1}{2}$. 
  We obtain $\typefn^{n+1}(x) = (\tent(v),\tent(\tfrac{1}{2})]$. 
\item[Case 2-1-2.] 
  If $\tent^n(x) > 1/2$ then $x_n \in (\frac{1}{2},u]$. 
  Thus, $\tent^{n+1}(x) =\tent(x_n) = \mu (1-x_n)$. 
  Accordingly, $\tent(\frac{1}{2}) = \mu (1-\frac{1}{2}) \geq \tent(x_n)  > \mu (1-u) = \tent(u)$ since $\frac{1}{2} < x_n \leq u$. 
  We obtain $\typefn^{n+1}(x) =  [\tent(u),\tent(\tfrac{1}{2}))$. 
\end{itemize}
\item[Case 2-2:] $u \leq \frac{1}{2}$. 
 Then, $x_n \leq \frac{1}{2}$. 
  Thus, $\tent^{n+1}(x) =\tent(x_n) = \mu x_n$. 
  Accordingly, $\tent(v) = \mu v < \tent(x_n)  \leq \mu u = \tent(u)$ since $v < x_n \leq u$. 
  We obtain $\typefn^{n+1}(x) = (\tent(v),\tent(u)]$. 
\item[Case 2-3:] $v \geq \frac{1}{2}$. 
 Then, $x_n > \frac{1}{2}$. 
  Thus, $\tent^{n+1}(x) =\tent(x_n) = \mu (1-x_n)$. 
  Accordingly, $\tent(v) = \mu (1-v) \geq \tent(x_n)  > \mu (1-u) = \tent(u)$ since $v < x_n \leq u$. 
 We obtain $\typefn^{n+1}(x) = [\tent(u),\tent(v))$. 
\end{itemize}
\end{proof}

\subsection{Proof of Lemma~\ref{lem:cond-prob}}\label{apx:cond-prob}
\begin{lemma}[Lemma~\ref{lem:cond-prob}]
 Let $X$ be a real-valued random variable drawn from $[0,1)$ uniformly at random. 
 Let ${\bf B}_{n+1} = B_1 \cdots B_{n+1} = \enc^{n+1}(X)$. 
 Let $\bitseq_n \in {\cal L}_n$. Then, 
\begin{align*}
 \Pr[ B_{n+1}= b \mid \enc^n(X) = \bitseq_n]
    = \frac{|\typefn(\bitseq_n b)|}{|\typefn(\bitseq_n 0)|+|\typefn(\bitseq_n 1)|}
\end{align*}
 holds for $b \in \{0,1\}$, where let  $|\typefn(\bitseq_n b)| = 0$ if $\bitseq_n b \not\in {\cal L}_{n+1}$. 
\end{lemma}
 As a preliminary step, 
   we prove the following lemma which {\em proves} that $\tent^n$ is a piecewise {\em linear} function, 
   which is almost trivial but we have not proven yet. 
\begin{lemma}\label{lem:uniform}
  $\tent^n(X)$ is uniformly distributed over $\typefn^n(X)$. 
\end{lemma}
\begin{proof}
 The proof is an induction on $n$. 
 We start with $n=1$. 
 Notice that $\Pr[X \in [0,\frac{1}{2})] = \Pr[X \in [\frac{1}{2},1)] = \frac{1}{2}$. 
 We consider two cases whether $X < \frac{1}{2}$ or not. 
 If $X<\frac{1}{2}$ then $\tent(X) = \mu X$. 
 Thus, 
\begin{align*}  
  \Pr \left[ \tent(X) < k \mid X \in [0,\tfrac{1}{2}) \right] 
    =\Pr \left[ \mu X < k \mid X \in [0,\tfrac{1}{2}) \right] 
    = \frac{\frac{k}{\mu}}{\frac{1}{2}}
    = \frac{k}{\frac{\mu}{2}}
\end{align*}
   holds for any $k \in [0,\frac{\mu}{2}) = \typefn(X)$, 
   which implies the claim in the case.
 Similarly, if $X \geq \frac{1}{2}$ then $\tent(X) = \mu (1-X)$. 
 Then, 
  $\Pr \left[ \tent(X) < k \mid X \in [\frac{1}{2},1)\right] 
    =\Pr \left[ \mu (1-X) < k \mid 1-X \in (0,\frac{1}{2}] \right] 
    = \frac{\frac{k}{\mu}}{\frac{1}{2}}
    = \frac{k}{\frac{\mu}{2}}$
   holds for any $k \in [0,\frac{\mu}{2}) = \typefn(X)$, 
   which implies the claim for $n=1$.

 Inductively assuming the claim for $n$, we prove it for $n+1$. 
 We here prove Case 1-1 in Lemma~\ref{lem:transition}, but other cases are similar. 
 Suppose $\typefn^n(X) = [v,u)$ and $\frac{1}{2} \in (v,u)$. 
 Let $X_n =\tent^n(X)$ for convenience.  
 If $X_n<\frac{1}{2}$ then $\tent(X_n) = \mu X_n$. 
 Then, 
\begin{align*}  
&  \Pr \left[ \tent^{n+1}(X) < k \mid X _n\in [v,\tfrac{1}{2}) \right] 
  =\Pr \left[ \tent(X_n) < k \mid X _n\in [v,\tfrac{1}{2}) \right] 
  = \Pr \left[ \mu X_n < k \mid X_n \in [v,\tfrac{1}{2}) \right] \\
& \hspace{2em}
     = \Pr \left[ X_n < \tfrac{k}{\mu} \mid X_n \in [v,\tfrac{1}{2}) \right] 
   = \frac{|[v,\frac{k}{\mu})|}{\left|[v,\tfrac{1}{2})\right|} 
   = \frac{|[\mu v,k)|}{\left|[\mu v,\tfrac{\mu}{2})\right|} 
   = \frac{|[\tent(v) ,k)|}{\left|[\tent(v), \tent(\tfrac{1}{2}))\right|} 
\end{align*}
 holds for $k \in [\tent(v),\tent(\frac{1}{2}))$, 
 where  $[\tent(v),\tent(\frac{1}{2})) = \typefn^{n+1}(X)$ holds 
  by the argument similar to the proof of Case 1-1-1 in Lemma~\ref{lem:transition}. 
 Thus,  $\tent^{n+1}(X)$ is uniformly distributed over $\typefn^{n+1}(X)$. 

 If $X_n \geq \frac{1}{2}$ then $\tent(X_n) = \mu (1-X_n)$. 
 Then, 
\begin{align*}  
&  \Pr \left[ \tent^{n+1}(X) < k \mid X _n\in [\tfrac{1}{2},u) \right] 
   =\Pr \left[ \tent(X_n) < k \mid X _n\in [\tfrac{1}{2},u) \right] \\
&\hspace{2em}
   = \Pr \left[ \mu (1-X_n) < k \mid 1-X_n \in (1-u,\tfrac{1}{2}] \right] 
   = \Pr \left[ 1-X_n < \tfrac{k}{\mu} \mid 1-X_n \in (1-u,\tfrac{1}{2}] \right] \\
&\hspace{2em}
 = \frac{|(1-u,\frac{k}{\mu}]|}{\left|(1-u,\tfrac{1}{2}]\right|} 
 = \frac{|(\mu(1-u),k]|}{\left|(\mu(1-u),\tfrac{\mu}{2}]\right|} 
 = \frac{|(\tent(u),k]|}{\left|(\tent(u),\tent(\tfrac{1}{2})]\right|} 
\end{align*}
 holds for $k \in (\tent(u),\tent(\frac{1}{2})]= \typefn^{n+1}(X)$. 
 We obtain Case 1-1. 
 It is not difficult to see that other cases are similar. 
\end{proof}

Next, we prove the following lemma. 
\begin{lemma}\label{lem:cond-prob-type}
 Let $\bitseq_n \in {\cal L}_n$. Then, 
\begin{align*}
   \frac{|\se_{n+1}(\bitseq_n b)|}{|\se_n(\bitseq_n)|}
   = \frac{|\typefn(\bitseq_n b)|}{|\typefn(\bitseq_n 0)|+|\typefn(\bitseq_n 1)|}
\end{align*}
 holds for $b=0,1$.  
\end{lemma}

\begin{proof}
 It is trivial if $\se_{n+1}(\bitseq_n 0)=\emptyset$ or $\se_{n+1}(\bitseq_n 1) = \emptyset$, 
  corresponding to Cases 1-2, 1-3, 2-2, 2-3 in Lemma~\ref{lem:transition}. 
 Consider Case 1-1 (Case 2-1 is similar).  
 Let $\typefn(\bitseq_n) = [v,u)$, and 
 let $x^* \in \se_n(\bitseq_n)$ satisfy $\tent^n(x^*) (= x^*_n) = \frac{1}{2}$.

Firstly we remark 
\begin{align}
   \frac{|\se_{n+1}(\bitseq_n 0)|}{|\se_{n+1}(\bitseq_n 0)|+|\se_{n+1}(\bitseq_n 1)|}
   = \Pr[ X < x^* \mid X \in \se_n(\bitseq_n)] 
\label{eq:cond-prob-type1}
\end{align}
 holds since $X$ is uniformly distributed over $[0,1)$. 
 
 Next, we claim that 
\begin{align}
  \Pr[ X_n < \tent^n(x^*) \mid X_n \in \typefn(\bitseq_n)] 
   = \frac{|\typefn(\bitseq_n b)|}{|\typefn(\bitseq_n 0)|+|\typefn(\bitseq_n 1)|}
\label{eq:cond-prob-type2}
\end{align}
 holds similarly from Lemma~\ref{lem:uniform}. 
In fact, 
 \begin{align*}
 \mbox{l.h.s.\ of \eqref{eq:cond-prob-type2}}
  =
  \Pr[ X_n < \tent^n(x^*) \mid X_n \in \typefn(\bitseq_n)] 
   = \frac{|[v,x^*)|}{|[v,u)|}
   = \frac{x^*-v}{u-v}
 \end{align*}
 holds by Lemma~\ref{lem:uniform}. 
 For the right hand side, we see that 
\begin{align}
  |\typefn(\bitseq_n 0)| &= \tent(\tfrac{1}{2})-\tent(v) = \mu (\tfrac{1}{2} - v) \label{eq:20230624a}\\
  |\typefn(\bitseq_n 1)| &= \tent(u) - \tent(\tfrac{1}{2}) = \mu (1-\tfrac{1}{2}) - \mu (1-u) = \mu (u-\tfrac{1}{2}) \label{eq:20230624b}
\end{align}
 holds. Thus 
 \begin{align*}
 \mbox{r.h.s.\ of \eqref{eq:cond-prob-type2}}
   = \frac{|\typefn(\bitseq_n b)|}{|\typefn(\bitseq_n 0)|+|\typefn(\bitseq_n 1)|}
   = \frac{\mu (\tfrac{1}{2} - v) }{\mu (\tfrac{1}{2} - v) +\mu (u-\tfrac{1}{2})} 
   = \frac{x^*-v}{u-v}
\end{align*}
 holds. We obtain \eqref{eq:cond-prob-type2}. 
 
 Finally, we observe 
\begin{align*}
\eqref{eq:cond-prob-type1}
 &= \Pr[ X < x^* \mid X \in \se_n(\bitseq_n)] \\
  &=  \Pr[ X < x^* \mid X_n \in \typefn(\bitseq_n)] 
  && (\mbox{since $x \in \se_n(\bitseq_n)$ iff $x_n \in \typefn(\bitseq_n)$})
  \\
 &=  \Pr[ X_n < \tent^n(x^*) \mid X_n \in \typefn(\bitseq_n)] 
  && (\mbox{since $x <x^*$ iff $x_n < \tent^n(x^*)$}) \\
 &=\eqref{eq:cond-prob-type2}
\end{align*}
 and we obtain the claim in the case. 
 Case 2-1 is similar. 
\end{proof}

\begin{proof}[Proof of Lemma~\ref{lem:cond-prob}]
 Recall  \eqref{eq:cond_prob2}, then 
\begin{align*}
\Pr[ B_{n+1}= b \mid \enc^n(X) = \bitseq_n]
       = \frac{|\se_{n+1}(\bitseq_n b)|}{|\se_{n+1}(\bitseq_n)|}
\end{align*}
 is trivial for $b=0,1$. 
 By Lemma~\ref{lem:cond-prob-type}, 
\begin{align*}
 \frac{|\se_{n+1}(\bitseq_n b)|}{|\se_{n+1}(\bitseq_n)|}
  = \frac{|\typefn(\bitseq_n b)|}{|\typefn(\bitseq_n 0)|+|\typefn(\bitseq_n 1)|}
\end{align*}
 holds  for $b=0,1$. 
 Now the claim is clear. 
\end{proof}

\begin{lemma}\label{lem:tasutomu}
\begin{align}
  |\typefn(\bitseq_n 0)|+|\typefn(\bitseq_n 1)| =\mu|\typefn(\bitseq_n)|
\end{align}
 holds for any $\bitseq_n \in {\cal L}_n$
\end{lemma}
\begin{proof}
 Immediate from \eqref{eq:20230624a} and  \eqref{eq:20230624b}. 
 \end{proof}
\end{document}